\documentclass[abstract]{scrartcl}
\usepackage{geometry} 
 
\usepackage{amsmath}
\usepackage{amsfonts}
\usepackage{amssymb}
\usepackage{amsthm}
\usepackage{graphicx}
\usepackage[english]{babel}

\usepackage[format=plain,labelfont=bf,textfont=it]{caption}

\usepackage{enumerate}

\usepackage{tikz}
\usetikzlibrary{arrows,cd}
\usepackage{ifthen}

\usepackage{hyperref}
\usepackage{doi}

\usepackage[numbers]{natbib}
\usepackage{authblk} 

\usepackage{empheq}

\usepackage{relsize}
\newcommand{\Asub}{{A_{c,\mathsmaller \blambda,\btau}^{-1}(\bt)}}

\renewcommand{\d}{\mathrm{d}}
\newcommand{\cF}{\mathcal{F}}
\newcommand{\cL}{\mathcal{L}}
\newcommand{\cT}{\mathcal{T}}
\newcommand{\var}[3]{\frac{\delta_{#1} {#2}}{\delta {#3}}}
\newcommand{\der}[2]{\frac{\partial {#1}}{\partial {#2}}}
\newcommand{\disc}{\mathrm{disc}}
\renewcommand{\mod}{\mathrm{mod}}
\newcommand{\miwa}{\mathrm{Miwa}}
\renewcommand{\O}{\mathcal{O}}

\newcommand{\C}{\mathbb{C}}
\newcommand{\R}{\mathbb{R}}
\newcommand{\Z}{\mathbb{Z}}

\newcommand{\q}[1]{ q^{\scriptscriptstyle [#1]} }
\newcommand{\Q}[1]{ Q^{\scriptscriptstyle [#1]} }

\newcommand{\bn}{\mathbf{n}}
\newcommand{\bt}{\mathbf{t}}
\newcommand{\blambda}{\boldsymbol{\lambda}}
\newcommand{\btau}{\boldsymbol{\tau}}

\newcommand{\fe}{\mathfrak{e}}
\newcommand{\fv}{\mathfrak{v}}

\DeclareMathOperator{\pr}{pr}
\DeclareMathOperator{\sgn}{sgn}
\DeclareMathOperator{\D}{D}

\newtheorem{thm}{Theorem}
\newtheorem{prop}[thm]{Proposition}
\newtheorem{lemma}[thm]{Lemma}

\theoremstyle{definition}
\newtheorem{defi}[thm]{Definition}

\title{Continuum limits of \\ pluri-Lagrangian systems \vspace{1cm}}
\author{Mats Vermeeren}
\date{\normalsize \today}

\affil{\normalsize \textit{Institut f\"ur Mathematik, MA 7-1, Technische Universit\"at Berlin, \\
Str.\@ des 17.\@ Juni 136, 10623 Berlin, Germany} \\
\texttt{vermeeren@math.tu-berlin.de}}
\date{}

\begin{document}

\maketitle

\begin{abstract}
\noindent
A pluri-Lagrangian (or Lagrangian multiform) structure is an attribute of integrability that has mainly been studied in the context of multidimensionally consistent lattice equations. It unifies multidimensional consistency with the variational character of the equations. An analogous continuous structure exists for integrable hierarchies of differential equations. We present a continuum limit procedure for pluri-Lagrangian systems. In this procedure the lattice parameters are interpreted as Miwa variables, describing a particular embedding in continuous multi-time of the mesh on which the discrete system lives. Then we seek differential equations whose solutions interpolate the embedded discrete solutions. The continuous systems found this way are hierarchies of differential equations. We show that this continuum limit can also be applied to the corresponding pluri-Lagrangian structures. We apply our method to the discrete Toda lattice and to equations H1 and Q1$_{\delta = 0}$ from the ABS list.
\end{abstract}

\vspace{1cm}

\section{Introduction}

A cornerstone of the theory of integrable systems is the idea that integrable equations come in families of compatible equations. In the continuous case these are hierarchies of differential equations with commuting flows. In the discrete case, in particular in the context of equations on quadrilateral graphs (\emph{quad equations}) this property is known as \emph{multidimensional consistency}. A classification of multidimensionally consistent quad equations was found by Adler, Bobenko, and Suris \cite{adler2003classification} and is often referred to as the ABS list. Additionally, many integrable equations can be derived from a variational principle. The \emph{Lagrangian multiform} or \emph{pluri-Lagrangian} formalism, which grew out of a beautiful insight by Lobb and Nijhoff \cite{lobb2009lagrangian}, combines these two aspects of integrability.%
\footnote{
	The author prefers the term ``pluri-Lagrangian'' over ``Lagrangian multiform'' because it is not the differential form that has a 	multiplicity or plurality to it, but rather its interpretation as a Lagrangian. 
	There is a minor distinction in how both names have been used in the literature: ``pluri-Lagrangian'' indicates that solutions are critical with respect to variations of the dependent variable on any fixed surface \cite{bobenko2015discrete,boll2014integrability,suris2016lagrangian}, whereas ``Lagrangian multiform'' is mostly used when one also requires criticality with respect to variations in the geometry of the surface \cite{king2017quantum,lobb2009lagrangian,xenitidis2011lagrangian,yoo2011discrete}. This distinction is not relevant to the present work.
}
The central idea in this notion is that action sum or integral can be taken on an arbitrary surface in a higher dimensional space. 

The discrete version of the the pluri-Lagrangian theory is more developed than the continuous one, and arguably more fundamental. Hence, connecting both sides could lead to a better understanding of the continuous theory. It is known that the lattice parameters of a discrete pluri-Lagrangian system may play the role of independent variables in a corresponding continuous system of non-autonomous differential equations, see e.g.\@ \cite{lobb2009lagrangian,xenitidis2011lagrangian}. This paper presents a different connection between discrete and continuous pluri-Lagrangian systems, where the continuous variables interpolate the discrete ones. The lattice parameters describe the size and shape of the mesh on which the discrete system lives, and thus they disappear in the continuum limit. The continuous systems found this way are hierarchies of autonomous differential equations. Pluri-Lagrangian structures for such hierarchies were studied independently of the discrete case in \cite{suris2016lagrangian}.

Some similar continuum limits can be found in the literature, for example in \cite{morosi1996continuous,morosi1998continuous,morosi1998continuous2,yoo2011discrete} and in particular in \cite{wiersma1987lattice}, where the lattice potential KdV equation is shown to produce the potential KdV hierarchy in a suitable limit. On the level of the pluri-Lagrangian structure, the problem is essentially that of interpolation of discrete variational systems by continuous Lagrangian systems. This was studied in \cite{vermeeren2015modified} because of its relevance in numerical analysis, in particular for backward error analysis of variational integrators. We will build on the ideas from that work to construct pluri-Lagrangian structures for hierarchies of differential equations that appear as continuum limits of lattice equations.

Section \ref{sec-pluri} contains a crash course on discrete and continuous pluri-Lagrangian systems. Section \ref{sec-miwa} provides an introduction to Miwa variables, which turn out to be a powerful tool for taking continuum limits. In fact, on the level of equations, this is the only tool required to obtain a continuous system. However, that leaves the question whether the resulting differential equations are integrable. In Section \ref{sec-laglim} we look at the Lagrangian side of the continuum limit; first we review the method of \cite{vermeeren2015modified} to take continuum limits of classical Lagrangian systems (regardless of their integrability), then we extend these ideas to pluri-Lagrangian systems. By recovering a pluri-Lagrangian structure on the continuous side, the question about integrability of the limit is settled in the affirmative. In Section \ref{sec-examples} we study several examples in detail.

\section{Pluri-Lagrangian systems}
\label{sec-pluri}

\subsection{Discrete pluri-Lagrangian systems}

Consider the lattice $\Z^N$ with basis vectors $\fe_1,\ldots,\fe_N$. To each lattice direction we associate a parameter $\lambda_i \in \C$. The equations we are interested in involve the values of a field $U: \Z^N \rightarrow \C$ on elementary squares in this lattice, or more generally, on $d$-dimensional \emph{plaquettes}. Such a plaquette is a $2^d$-tuple of lattice points that form an elementary hypercube. We denote it by
\[ \square_{i_1,\ldots,i_d}(\bn) = \left\{ \bn + \varepsilon_1 \fe_{i_1} + \ldots + \varepsilon_d \fe_{i_d} \,\Big|\, \varepsilon_k \in \{0,1\} \right\} \subset \Z^N, \]
where $\bn = (n_1, \ldots, n_N)$. Plaquettes are considered to be oriented; an odd permutation of the directions $i_1,\ldots,i_d$ reverses the orientation of the plaquette. We will write $ U(\square_{i_1,\ldots,i_d}(\bn))$ for the $2^d$-tuple
\[ U(\square_{i_1,\ldots,i_d}(\bn)) = \Big( U(\bn), U(\bn + \fe_{i_1}), U(\bn + \fe_{i_2}),  \ldots, U(\bn + \fe_{i_1} \ldots + \fe_{i_d}) \Big) . \]
Occasionally we will also consider the corresponding ``filled-in'' hypercubes in $\R^N$,
\[ \blacksquare_{i_1,\ldots,i_d}(\bn) = \left\{ \bn + \alpha_1 \fe_{i_1} + \ldots + \alpha_d \fe_{i_d} \,\Big|\, \alpha_k \in [0,1] \right\} \subset \R^N , \]
on which we consider the orientation defined by the volume form $\d t_{i_1} \wedge \ldots \wedge \d t_{i_d}$.

The role of a Lagrange function is played by a discrete $d$-form 
\[ L(U(\square_{i_1,\ldots,i_d}(\bn) ), \lambda_{i_1}, \ldots, \lambda_{i_d}), \] 
i.e.\@ a function of the values of the field $U: \Z^N \rightarrow \C$ on a plaquette and of the corresponding lattice parameters, where 
\[ L \!\left( U \!\left( \square_{\sigma(i_1),\ldots,\sigma(i_d)}(\bn) \right), \lambda_{\sigma(i_1)}, \ldots, \lambda_{\sigma(i_d)} \right) = \sgn(\sigma) L(\square_{i_1,\ldots,i_d}(\bn), \lambda_{i_1}, \ldots, \lambda_{i_d}) \]
for any permutation $\sigma$ of $i_1, \ldots, i_d$. 

Consider a discrete $d$-surface $\Gamma = \{\square_\alpha\}$ in the lattice, i.e.\@ a set of $d$-dimensional plaquettes, indexed by a parameter $\alpha$, such that the union of the corresponding filled-in plaquettes  $\bigcup_\alpha \blacksquare_\alpha$ is an oriented topological $d$-manifold (possibly with boundary). The action over $\Gamma$ is given by
\begin{equation}\label{discact}
S_\Gamma = \sum_{\square_{i_1,\ldots,i_d}(\bn) \in \Gamma} L(U(\square_{i_1,\ldots,i_d}(\bn)) ,\lambda_{i_1}, \ldots, \lambda_{i_d}) .
\end{equation}
The field $U$ is a solution to the \emph{pluri-Lagrangian problem} if it is a critical point of $S_\Gamma$ (with respect to variations that are zero on the boundary of $\Gamma$) for all discrete $d$-surfaces $\Gamma$ simultaneously.

\begin{figure}[t]
\centering
\begin{tikzpicture}[scale=(0.75)]
\def\a{.8}
\def\b{.6}
\newcommand{\tikzsquare}[4]{
	\ifthenelse{#4=1} {
	\filldraw[fill=black!50!, draw=black, fill opacity=0.8]  (#1+\a*#3,#2+\b*#3) -- (#1+\a*#3+\a,#2+\b*#3+\b) -- (#1+\a*#3+\a,#2+\b*#3+1+\b) -- (#1+\a*#3,#2+\b*#3+1) -- cycle;
	} {}
	\ifthenelse{#4=2} {
	\filldraw[fill=black!10!, draw=black, fill opacity=0.8] (#1+\a*#3,#2+\b*#3) -- (#1+\a*#3+\a,#2+\b*#3+\b) -- (#1+\a*#3+1+\a,#2+\b*#3+\b) -- (#1+\a*#3+1,#2+\b*#3) -- cycle;
	} {}
	\ifthenelse{#4=3} {
	\filldraw[fill=black!30!, draw=black, fill opacity=0.8](#1+\a*#3,#2+\b*#3) -- (#1+\a*#3,#2+\b*#3+1) -- (#1+\a*#3+1,#2+\b*#3+1) -- (#1+\a*#3+1,#2+\b*#3) -- cycle;
	} {}
}

\tikzsquare{-1}{0}{0}{2};
\tikzsquare{-1}{0}{1}{2};
\tikzsquare{-1}{0}{2}{2};
\tikzsquare{-1}{0}{3}{2};
\tikzsquare{0}{0}{0}{2};
\tikzsquare{0}{0}{3}{2};
\tikzsquare{1}{0}{0}{2};
\tikzsquare{1}{0}{1}{2};
\tikzsquare{1}{0}{3}{2};
\tikzsquare{2}{0}{0}{2};
\tikzsquare{2}{0}{1}{2};
\tikzsquare{2}{0}{2}{2};
\tikzsquare{2}{0}{3}{2};

\tikzsquare{0}{0}{1}{1};
\tikzsquare{0}{0}{2}{1};
\tikzsquare{1}{1}{2}{1};
\tikzsquare{0}{0}{3}{3};
\tikzsquare{1}{0}{3}{3};
\tikzsquare{1}{1}{3}{3};

\tikzsquare{0}{0}{1}{3};
\tikzsquare{0}{1}{1}{2};
\tikzsquare{0}{1}{2}{2};
\tikzsquare{1}{0}{1}{1};
\tikzsquare{1}{0}{2}{3};
\tikzsquare{1}{1}{2}{3};
\tikzsquare{1}{2}{2}{2};
\tikzsquare{2}{0}{2}{1};
\tikzsquare{2}{1}{2}{1};
\end{tikzpicture}
\caption{Visualization of a discrete $2$-surface in $\mathbb{Z}^3$.}
\end{figure}
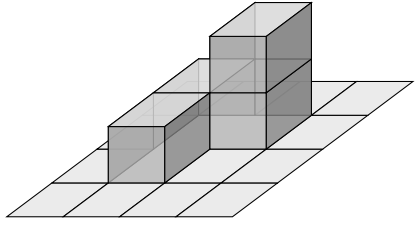

For $d = 1$, Equation \eqref{discact} reads
\[ S_\Gamma = \sum_{\{\bn, \bn + \fe_i\} \in \Gamma} L( U(\bn), U(\bn + \fe_i),\lambda_i) . \]
The Euler-Lagrange equations at general elementary corners,
\[ \der{}{U(\bn)} \Big( L( U(\bn \pm \fe_i), U(\bn),\lambda_i) + L( U(\bn), U(\bn \pm \fe_j),\lambda_j) \Big) = 0, \]
are necessary and sufficient conditions for $U$ to be a solution to the pluri-Lagrangian problem.

For $d = 2$, Equation \eqref{discact} becomes
\[ S_\Gamma = \sum_{\{\bn, \bn + \fe_i, \bn + \fe_j, \bn + \fe_i + \fe_j\} \in \Gamma} L( U(\bn), U(\bn + \fe_i), U(\bn + \fe_j), U(\bn + \fe_i + \fe_j),\lambda_i,\lambda_j) . \]
Since every surface can be constructed out of corners of cubes, it is sufficient to determine the Euler-Lagrange equations on these elementary building blocks. They are
\begin{alignat*}{3}
\der{}{U} \Big( && L( U, U_i, U_j, U_{ij},\lambda_i,\lambda_j) 
&+ L( U, U_j, U_k, U_{jk},\lambda_j,\lambda_k) \\[-2mm]
&& &+ L( U, U_k, U_i, U_{ik},\lambda_k,\lambda_i) & \Big) &= 0 , 
\\
\der{}{U_i} \Big( && L( U, U_i, U_j, U_{ij},\lambda_i,\lambda_j)
&- L( U_i, U_{ij}, U_{ik}, U_{ijk},\lambda_j,\lambda_k) \\[-2mm]
&& &+ L( U, U_k, U_i, U_{ik},\lambda_k,\lambda_i) & \Big) &= 0 , 
\\
\der{}{U_{ij}} \Big( && L( U, U_i, U_j, U_{ij},\lambda_i,\lambda_j)
&- L( U_i, U_{ij}, U_{ik}, U_{ijk},\lambda_j,\lambda_k) \\[-2mm]
&& &- L( U_j, U_{jk}, U_{ij}, U_{ijk},\lambda_k,\lambda_i) & \Big) &= 0 , 
\\
\der{}{U_{ijk}} \Big( && - L( U_k, U_{ik}, U_{jk}, U_{ijk},\lambda_i,\lambda_j) 
&- L( U_i, U_{ij}, U_{ik}, U_{ijk},\lambda_j,\lambda_k) \\[-2mm]
&& &- L( U_j, U_{jk}, U_{ij}, U_{ijk},\lambda_k,\lambda_i) & \Big) &= 0 .
\end{alignat*}
These corner equations are necessary and sufficient conditions for $U$ to be a solution to the pluri-Lagrangian problem. Often, $L$ can be written in a three-leg form
\[ L(U, U_i, U_j, U_{ij}, \lambda_i,\lambda_j) = A(U, U_i, \lambda_i) - A(U, U_j, \lambda_j) + B(U_i, U_j, \lambda_i - \lambda_j) , \]
which renders the first and last corner equations trivial. In particular, this is the case for all equations from the ABS list.

For more details, we refer to \cite{lobb2009lagrangian}, \cite{boll2014integrability}, \cite[Chapter 12]{hietarinta2016discrete}, and the references therein.

\subsection{Continuous pluri-Lagrangian systems}

In the continuous case, the lattice is replaced by a space $\R^N$, which we refer to as \emph{multi-time}. The Lagrangian in this context is a differential $d$-form
\begin{equation}\label{d-form}
\cL = \sum_{1 \leq i_1<\ldots<i_d \leq N} \cL_{i_1,\ldots,i_d}[u] \, \d t_{i_1} \wedge \ldots \wedge \d t_{i_d} ,
\end{equation}
where the square brackets denote dependence on the field $u: \R^N \rightarrow \C$ and an arbitrary number of its partial derivatives. We will always use lower case letters to denote continuous fields, as opposed to the upper case letters used for discrete fields. The field $u$ solves the \emph{pluri-Lagrangian problem} if for any $d$-dimensional submanifold $\Gamma$ of $\R^N$ it is a critical point of the action
\[ S_\Gamma = \int_\Gamma \cL \]
with respect to variations that are zero near the boundary of $\Gamma$. An infinite hierarchy of integrable differential equations is described by formal $d$-form in infinite dimensions,
\[ \cL = \sum_{1 \leq i_1<\ldots<i_d < \infty} \cL_{i_1,\ldots,i_d}[u] \, \d t_{i_1} \wedge \ldots \wedge \d t_{i_d} , \]
where for any $N \in \mathbb{N}$ its restriction \eqref{d-form} to $\R^N$ is a pluri-Lagrangian $d$-form for the part of the hierarchy involving derivatives with respect to $t_1,\ldots,t_N$ only.

The \emph{multi-time Euler-Lagrange equations}, which characterize solutions to the pluri-Lagrangian problem, were derived in \cite{suris2016lagrangian} for $d=1$ and $d=2$. The main idea of that derivation is to approximate any given smooth $d$-surface by a \emph{stepped surface}, a piecewise flat surface, the pieces of which are shifted sections of coordinate planes. Analogous to the discrete case, it is sufficient to look at the elementary building blocks of stepped surfaces.
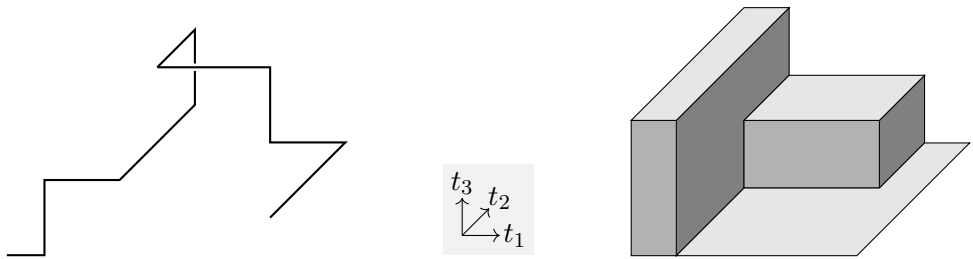
\begin{figure}[t]
\centering
\begin{tikzpicture}[scale=(0.5)]
\draw[thick] (0,0) -- (1,0) -- (1,2) -- (3,2) -- (5,4) -- (5,6) -- (4,5);
\draw[white, line width=1mm] (4.5,5) -- (5.5,5);
\draw[thick] (4,5) -- (7,5) -- (7,3) -- (9,3) -- (7,1);
\end{tikzpicture}
\hspace{1cm}
\begin{tikzpicture}[scale=(0.5)]
\fill[black!5!] (-.5,-.5) -- (1.9,-.5) -- (1.9,1.9) -- (-.5,1.9);
\draw[->] (0,0) -- (1,0);
\draw[->] (0,0) -- (0,1);
\draw[->] (0,0) -- (.72,.72);
\node at (1.4,0) {$t_1$};
\node at (1,1) {$t_2$};
\node at (0,1.4) {$t_3$};
\end{tikzpicture}
\hspace{1cm}
\begin{tikzpicture}[scale=(0.3)]
\filldraw[fill=black!50!, draw=black] (-3,-3) -- (-3,3) -- (2,8) -- (2,5) -- (0,3) -- (0,0) -- cycle;
\filldraw[fill=black!50!, draw=black] (6,0) -- (6,3) -- (8,5) -- (8,2) -- cycle;

\filldraw[fill=black!30!, draw=black] (0,3) -- (6,3) -- (6,0) -- (0,0) -- cycle;
\filldraw[fill=black!30!, draw=black] (-3,-3) -- (-5,-3) -- (-5,3) -- (-3,3) -- cycle;

\filldraw[fill=black!10!, draw=black] (-3,-3) -- (5,-3) -- (10,2) -- (8,2) -- (6,0) -- (0,0) -- cycle;
\filldraw[fill=black!10!, draw=black] (0,8) -- (-5,3) -- (-3,3) -- (2,8) -- cycle;
\filldraw[fill=black!10!, draw=black] (0,3) -- (2,5) -- (8,5) -- (6,3) -- cycle;
\end{tikzpicture}
\caption{A stepped curve (left) and a stepped $2$-surface (right) in $\R^3$}
\end{figure}

In order to state the multi-time Euler-Lagrange equations we introduce a multi-index notation for partial derivatives. An \emph{$N$-index} $I$ is an $N$-tuple of non-negative integers. There is a natural bijection between $N$-indices and partial derivatives of $u:\R^N \rightarrow \C$. We denote by $u_I$ the mixed partial derivative of $u$, where the number of derivatives with respect to each $t_i$ is given by the entries of $I$. Note that if $I = (0,\ldots,0)$,  then $u_I = u$. 

We will often denote a multi-index suggestively by a string of $t_i$-variables, but it should be noted that this representation is not always unique. For example,
\[ t_1 = (1,0,\ldots,0), \qquad t_N = (0,\ldots,0,1), \qquad t_1 t_2 = t_2 t_1 = (1,1,0,\ldots,0) . \]
In this notation, we will also make use of exponents to compactify the expressions, for example
\[ t_2^3 = t_2 t_2 t_2 = (0,3,0,\ldots,0). \]
The notation $I t_j$ should be interpreted as concatenation in the string representation, hence it denotes the multi-index obtained from $I$ by increasing the $j$-th entry by one. Finally, if the $j$-th entry of $I$ is nonzero we say that $I$ contains $t_j$, and write $I \ni t_j$.

For $d = 1$ the multi-time Euler-Lagrange equations are
\begin{subequations}
\begin{alignat}{3}
& \var{i}{\cL_i}{u_I} = 0 & \forall I \not\ni t_i ,  \label{EL11}\\
& \var{i}{\cL_i}{u_{It_i}} = \var{j}{\cL_j}{u_{It_j}} & \forall I , \label{EL12}
\end{alignat}
\end{subequations}
where $\var{i}{}{u_I}$ denotes a variational derivative in the $t_i$-direction,
\begin{align*}
\var{i}{}{u_I} &= \sum_{k = 0}^\infty (-1)^k \D_{t_i}^k \der{}{u_{It_i^k}} 
= \der{}{u_I} - \D_{t_i} \der{}{u_{It_i}} + \D_{t_i}^2 \der{}{u_{It_it_i}} - \ldots ,
\end{align*}
and $\D_{t_i} = \frac{\d}{\d t_i}$. Equation \eqref{EL11} is obtained from the straight parts of a stepped curve, Equation \eqref{EL12} from the corners.

For $d = 2$ the multi-time Euler-Lagrange equations are
\begin{subequations}
\begin{alignat}{3}
& \var{ij}{\cL_{ij}}{u_I} = 0 & \forall I \not\ni t_i, t_j , \label{EL21} \\
& \var{ij}{\cL_{ij}}{u_{It_j}} = \var{ik}{\cL_{ik}}{u_{It_k}} & \forall I \not\ni t_i , \label{EL22} \\
& \var{ij}{\cL_{ij}}{u_{It_it_j}} + \var{jk}{\cL_{jk}}{u_{It_jt_k}} + \var{ki}{\cL_{ki}}{u_{It_kt_i}} = 0 & \forall I , \label{EL23}
\end{alignat}
\end{subequations}
where
\[ \var{ij}{}{u_I} = \sum_{k = 0}^\infty \sum_{\ell = 0}^\infty (-1)^{k+\ell} \D_{t_i}^k \D_{t_j}^\ell \der{}{u_{I t_i^k t_j^\ell}} .\]
Equation \eqref{EL21} is obtained from the flat pieces of a stepped surface, Equation \eqref{EL22} from the edges, and  Equation \eqref{EL23} from the corners.

Note that there is no analogue of the lattice parameters in the continuous pluri-Lagrangian framework, but of course it is possible to consider parameter-dependent Lagrangians in the continuous case as well. One way of connecting the discrete and continuous cases is to consider the lattice parameters as independent variables of the continuous system and the discrete independent variables as parameters in the continuous system. This leads to a parameter-dependent non-autonomous PDE, known as the \emph{generating PDE}, which is discussed for example in \cite{lobb2009lagrangian}, \cite{nijhoff2000schwarzian} and \cite{xenitidis2011lagrangian}. We will briefly come back to it at the end of this paper.

The main goal of this work is to present a continuum limit procedure for pluri-Lagrangian systems. Instead of switching the roles of parameters and independent variables, we assume that the discrete system lives on a mesh embedded in $\R^N$, which is described by the lattice parameters. We then seek a continuous system which interpolates the lattice system.

\section{Miwa variables}
\label{sec-miwa}

To motivate our approach to the continuum limit, we start by considering the opposite direction.%
\footnote{The author is grateful to Yuri Suris for suggesting the motivation presented here.}
The problem of integrable discretization has been studied at impressive length in the monograph \cite{suris2003problem}. Let us briefly summarize the ``recipe'' for discretizing Toda-type systems from Section 2.9 of that work. It starts from an integrable ODE with a Lax representation of the form 
\begin{equation}\label{lax-cont}
 L_{t} = \left[ L, \pi_+(f(L)) \right]
\end{equation}
in a matrix Lie algebra $\mathfrak{g} = \mathfrak{g}_+ \oplus \mathfrak{g}_-$, where $\pi_+$ denotes projection onto $\mathfrak{g}_+$. Here $L$ denotes the Lax matrix, not to be confused with a Lagrangian, and $f:\mathfrak{g} \rightarrow \mathfrak{g}$ is an $\mathrm{Ad}$-covariant function. Such an equation is part of an integrable hierarchy, given by
\begin{equation}\label{lax-hier}
L_{t_k} = \left[ L, \pi_+ \big( f(L)^k \big) \right].
\end{equation}
A related integrable difference equation can be formulated in the corresponding Lie group $G$, with subgroups $G_+$ and $G_-$ having Lie algebras $\mathfrak{g}_+$ and $\mathfrak{g}_-$ respectively. Any element $x \in G$ close to the unit $\mathrm{Id} \in G$ can be factorized as $x = \Pi_+(x) \Pi_-(x)$, where $\Pi_\pm(x) \in G_\pm$. The difference equation is given by
\begin{equation}\label{lax-disc}
\widetilde{L} = \Pi_+(F(L))^{-1} \, L \, \Pi_+(F(L)),
\end{equation}
where the tilde $\widetilde{\cdot}$ denotes a discrete time step and 
\[ F(L) = \mathrm{Id} + \lambda f(L) \]
for some small parameter $\lambda$.

Solutions of the differential equation \eqref{lax-cont} are given by
\[ L(t) = \Pi_+\!\left( e^{t f(L_0)} \right)^{-1} L_0 \, \Pi_+\!\left( e^{t f(L_0)} \right) . \]
A simultaneous solution to the whole hierarchy \eqref{lax-hier} takes the form
\begin{equation}\label{recipe-c}
L(t_1,t_2,\ldots) = \Pi_+\!\left( e^{t_1 f(L_0) + t_2 f(L_0)^2 + \ldots} \right)^{-1} L_0 \, \Pi_+\!\left( e^{t_1 f(L_0) + t_2 f(L_0)^2 + \ldots} \right) .
\end{equation}
A solution of the discretization \eqref{lax-disc} is given by
\begin{align}
L(n) 
&= \Pi_+\!\left( F^n(L_0) \right)^{-1} L_0 \, \Pi_+\!\left( F^n(L_0) \right) \notag\\
&= \Pi_+\!\left( e^{n\log( 1 + \lambda f(L_0) ) } \right)^{-1} L_0 \, \Pi_+\!\left( e^{n\log( 1 + \lambda f(L_0) ) } \right) \notag \\
&= \Pi_+\!\left( e^{n \lambda f(L_0) - \frac{n}{2} \lambda^2 f(L_0)^2 + \ldots } \right)^{-1} L_0 \, \Pi_+\!\left( e^{n \lambda f(L_0) - \frac{n}{2} \lambda^2 f(L_0)^2 + \ldots } \right)\label{recipe-d} .
\end{align}
Comparing equations \eqref{recipe-c} and \eqref{recipe-d}, it is natural to identify a discrete step $n \mapsto n+1$ with a time shift
\[ (t_1, t_2, \ldots, t_i, \ldots) \mapsto \left(t_1 + \lambda, t_2 - \frac{\lambda^2}{2}, \ldots, t_i + (-1)^{i+1} \frac{\lambda^i}{i}, \ldots \right) . \]
This gives us a map from the discrete space $\Z^N(n_1, \ldots,n_N)$ into the continuous multi-time $\R^N(t_1,\ldots,t_N)$. We associate a parameter $\lambda_i$ with each lattice direction and set
\[ t_i = (-1)^{i+1} \left( n_1 \frac{\lambda_1^i}{i} + \ldots + n_N \frac{\lambda_N^i}{i} \right). \]
Note that a single step in the lattice (changing one $n_j$) affects all the times $t_i$, hence we are dealing with a very skew embedding of the lattice. We will also consider a slightly more general correspondence,
\begin{equation}\label{miwa}
t_i = (-1)^{i+1} \left(n_1 \frac{c \lambda_1^i}{i} + \ldots + n_N \frac{c \lambda_N^i}{i} \right) + \tau_i,
\end{equation}
for constants $c, \tau_1, \ldots, \tau_N$ describing a scaling and a shift of the lattice. The variables $n_j$ and $\lambda_j$ are known in the literature as Miwa variables and have their origin in \cite{miwa1982hirota}. In the present work we will call the $n_j$ \emph{discrete coordinates}, the $\lambda_j$ \emph{lattice parameters} and the  $t_i$ \emph{continuous coordinates} or \emph{times}. We will call Equation \eqref{miwa} the \emph{Miwa correspondence}. Let $\blambda = (\lambda_1, \ldots, \lambda_N)$ and consider the $N \times N$ matrix
\[ M_{\blambda} = \left( (-1)^{i+1} \frac{\lambda_j^i}{i} \right)_{i,j = 1}^N . \]
Then we can write the Miwa correspondence as
\[ \bt = c M_{\blambda} \bn + \btau, \]
where $\bt = (t_1, \ldots, t_N)^T$, $\bn = (n_1, \ldots, n_N)^T$, and $\btau = (\tau_1, \ldots, \tau_N)^T$.
In other words, we consider the mesh $\Z^N$ under the affine transformation
\begin{equation}\label{miwaA}
A_{c,\blambda,\btau} : \R^N \rightarrow \R^N: \bt \mapsto c M_{\blambda} \bt + \btau.
\end{equation}

We will use the Miwa correspondence \eqref{miwa}--\eqref{miwaA} even if the discrete system is not generated by the recipe described above. In many cases one can justify this in a similar way by considering \emph{plane wave factors}, solutions of the linearized system. For more on this perspective, see e.g.\@ \cite{nijhoff1995discrete,nijhoff1983direct,wiersma1987lattice} and \cite[Chapter 5]{hietarinta2016discrete}. 

For a completely different motivation for Miwa variables, note that for $N$ distinct parameter values $\lambda_1,\ldots,\lambda_N$ the corresponding vectors
\[ \nu(\lambda) = \left(c\lambda, -\frac{c\lambda^2}{2}, \ldots, - (-1)^N \frac{c\lambda^N}{N} \right) \]
are linearly independent. Up to projective transformations, $\nu$ is the only curve with that property. It is known as the \emph{rational normal curve} \cite{harris1992algebraic}.

To perform the continuum limit of a difference equation involving $U: \Z^N \rightarrow \C$, we associate to it a function $u: \R^N \rightarrow \C$ that interpolates it:
\[ U(\bn) = u( A_{c,\blambda,\btau}(\bn)) \qquad \forall \bn \in \Z^N, \]
where $A_{c,\blambda,\btau}$ is the Miwa embedding as given by Equation \eqref{miwaA}. We denote the shift of $U$ in the $i$-th lattice direction by $U_i$. If $U(\bn) = u(t_1,\ldots,t_N)$, it is given by
\[
U_i = U(\bn + \fe_i) =
u \!\left( t_1 + c \lambda_i, t_2 - \frac{c \lambda_i^2}{2}, \ldots, t_n - (-1)^N \frac{c \lambda_i^N}{N} \right), \]
which we can expand as a power series in $\lambda_i$. The difference equation thus turns into a power series in the lattice parameters. If all goes well, its coefficients will define differential equations that form an integrable hierarchy. Every solution of the continuous hierarchy will interpolate solutions of the discrete system, but there might be discrete solutions that do not correspond to any continuous solution. Examples can be found in Section \ref{sec-examples}.

Note that such a procedure is strictly speaking not a continuum \emph{limit}; sending $\lambda_i \rightarrow 0$ would only leave the leading order term of the power series. A more precise formulation is that the continuous $u$ interpolates the discrete $U$ for sufficiently small values of $\lambda_i$, where $U$ is defined on a mesh that is embedded in $\R^N$ using the Miwa correspondence. Since $\lambda_i$ is still assumed to be small, it makes sense to think of the outcome as a limit, but it is important to keep in mind that higher order terms should not be disregarded.

\section{Continuum limits of Lagrangian forms}
\label{sec-laglim}

\subsection{Modified Lagrangians in the classical variational problem}
\label{sec-straight-limit}

In \cite{vermeeren2015modified} we performed a continuum limit on Lagrangian systems in the context of variational integrators for ODEs. Given a discrete Lagrangian, we constructed a continuous \emph{modified Lagrangian} whose critical curves interpolate solutions of the discrete problem. A similar approach can be used in the context of pluri-Lagrangian systems, but first we present the relevant ideas in the context of the classical variational formulation of a P$\Delta$E, where $N = d$. We use an orthogonal lattice with a fixed mesh size $h$. In Section \ref{sec-pluri-limit} we will consider the pluri-Lagrangian problem and embed the lattice according to the Miwa correspondence.

In the classical discrete variational principle we consider elementary plaquettes of full dimension, so it is sufficient to label them only by position, leaving out the subscripts denoting the direction. We consider Lagrangians $L_\disc(\square(\bn), h)$ depending on the values of the field $U: \Z^d \rightarrow \C$ on a plaquette $\square(\bn)$ and on the mesh size $h$. As before, we denote lattice shifts by subscripts:
\[ U =  U(\bn), \qquad U_i = U( \bn + \fe_{i}), \qquad U_{-i} = U( \bn - \fe_{i}), \qquad U_{ij} = U( \bn + \fe_{i} + \fe_{j}), \quad \cdots . \]

We identify points of a discrete solution with mesh size $h$ with evaluations of an interpolating field $\widetilde{u}: \R^d \rightarrow \C$,
\begin{equation}\label{orthogonal}
 U(n_1,\ldots,n_d) = \widetilde{u}(h n_1, \ldots, h n_d) .
\end{equation}
Using a Taylor expansion we can write the discrete Lagrangian $L_\disc(\square(\bn), h)$ as a function of the interpolating field $\widetilde{u}$ and its derivatives. This defines $\widetilde{\cL}_\disc$,
\begin{align*}
&\widetilde{\cL}_\disc([\widetilde{u}], h) = L_\disc\!\left( \bigg\{ \widetilde{u} + h \sum_{k=1}^d \varepsilon_k \, \widetilde{u}_{t_k} + \frac{h^2}{2} \sum_{k=1}^d \sum_{\ell=1}^d \varepsilon_k \varepsilon_\ell \, \widetilde{u}_{t_k t_\ell} + \ldots \ \bigg|\ \varepsilon_k \in \{0,1\} \bigg\}, h \right) ,
\end{align*} 
where the square brackets denote dependence on $\widetilde{u}$ and any number of its partial derivatives.

So far we have only written the discrete Lagrangian as a function of the continuous field. The corresponding action is still a sum,
\begin{align*}
S(U,h) 
&= \sum_{\bn \in \Z^d} L_\disc(U(\square(\bn)),h)  
= \sum_{\bn \in \Z^d} \widetilde{\cL}_\disc([\widetilde{u}(\bn)], h) ,
\end{align*} 
which is why we sill use the subscript $\disc$. We want to write the action as an integral. This can be done using the Euler-Maclaurin formula, which relates sums to integrals \cite[Eq. 23.1.30]{abramowitz1964handbook}:
\begin{align*} 
\sum_{k=0}^{m-1} F(a+kh) 
&= \frac{1}{h} \int_a^{a+mh} F(t) \,\d t + \sum_{i=1}^\infty h^{i-1}\frac{ B_i}{i!} \left( F^{(i-1)}(a+mh) - F^{(i-1)}(a) \right) \\
&= \frac{1}{h} \int_a^{a+mh} \left( \sum_{i=0}^\infty h^i \frac{B_i}{i!} F^{(i)}(t) \right) \d t,
\end{align*} 
where $B_i$ denote the Bernoulli numbers $1, -\frac{1}{2}, \frac{1}{6}, 0, -\frac{1}{30},0,\cdots$. For any given field $\widetilde{u}$ and lattice parameter $h$ we apply the Euler-Maclaurin formula to $\widetilde{\cL}_\disc([\widetilde{u}], h)$ in each of the lattice directions. We obtain the \emph{modified Lagrangian}
\[\cL_\mod([\widetilde{u}],h) = \sum_{i_1,\ldots,i_d=0}^\infty h^{i_1 + \ldots + i_d} \frac{B_{i_1} \ldots B_{i_d}}{i_1! \ldots i_d!} \D_{t_1}^{i_1} \ldots \D_{t_d}^{i_d} \widetilde{\cL}_\disc([\widetilde{u}],h). \]
The power series in the Euler-Maclaurin Formula generally does not converge. The same is true for the series defining $\cL_\mod$. As a formal power series, it satisfies
\begin{equation}\label{meshed-action}
S(U,h) = \int_{\R^d} \cL_\mod([\widetilde{u}(\bt)],h) \,\d\bt ,
\end{equation}
where $\d \bt = \d t_1 \wedge \ldots \wedge \d t_d$ and $U$ is related to $\widetilde{u}$ by Equation \eqref{orthogonal}. This property also holds locally,
\begin{equation}\label{meshed}
 L_\disc(U(\square(\bn)),h) = \int_{\blacksquare(\bn)} \cL_\mod([\widetilde{u}(\bt)],h) \,\d\bt.
\end{equation}
To make these statements precise, we introduce a truncated version of the modified Lagrangian, which only includes powers of $h$ up to a certain degree $k$,
\[\cL_{\mod,k}([\widetilde{u}],h) = \sum_{i_1+\ldots+i_d\leq k} h^{i_1 + \ldots + i_d} \frac{B_{i_1} \ldots B_{i_d}}{i_1! \ldots i_d!} \D_{t_1}^{i_1} \ldots \D_{t_d}^{i_d} \widetilde{\cL}_\disc([\widetilde{u}],h). \]
With $\cL_{\mod,k}$ instead of $\cL_\mod$, Equations \eqref{meshed-action} and \eqref{meshed} hold with a defect of order $\O(h^{k+1})$.

\subsection{From discrete to continuous pluri-Lagrangian structures}
\label{sec-pluri-limit}

In the pluri-Lagrangian context we consider a discrete Lagrangian $d$-form in a higher dimensional lattice $\Z^N$, $N > d$. Consider $N$ pairwise distinct lattice parameters $\blambda = (\lambda_1, \ldots, \lambda_N)$ and denote by $\fe_1, \ldots, \fe_N$ the unit vectors in the lattice $\mathbb{Z}^N$. The differential $c M_{\blambda} $ of the Miwa correspondence $A_{c,\blambda,\btau}$, introduced in Equation \eqref{miwaA}, maps them to linearly independent vectors in $\R^N$:
\[ \fe_i \mapsto \fv_i = c M_{\blambda} \fe_i = \left(  c \lambda_i, - \frac{c \lambda_i^2}{2}, \ldots , (-1)^{N+1} \frac{c \lambda_i^N}{N} \right)^T . \]

We calculate the modified Lagrangian in the transformed coordinate system, starting from a discrete Lagrangian $d$-form  $L_\disc$. For the Miwa embedding of the lattice, $\cL_\disc$ is defined by
\begin{align}
&\cL_\disc([u], \lambda_1, \ldots, \lambda_d) \Big|_\bt \notag\\
&= L_\disc\!\left( \left\{ u(\bt + \varepsilon_1 \fv_1 + \ldots + \varepsilon_d \fv_d) \ \Big|\ \varepsilon_k \in \{0,1\} \right\}, \lambda_1, \ldots, \lambda_d \right) \notag\\
&= L_\disc\!\left( \bigg\{ u + \sum_{k=1}^d \varepsilon_k \partial_k u + \frac{1}{2} \sum_{k=1}^d \sum_{\ell=1}^d \varepsilon_k \varepsilon_\ell \partial_k \partial_\ell u + \ldots \ \bigg|\ \varepsilon_k \in \{0,1\} \bigg\}, \lambda_1, \ldots, \lambda_d \right) \Bigg|_\bt, \label{Ldisc}
\end{align}
where now the differential operators correspond to the lattice directions under the Miwa correspondence,
\[ \partial_k = \sum_{j=1}^N (-1)^{j+1} \frac{c \lambda_k^j}{j} \D_{t_j}. \]
Interpreted as a vector field, $\partial_k = \fv_k$ is the pushforward by $A_{c,\mathsmaller \blambda,\btau}$ of $\D_{t_k} = \fe_k$.

\begin{figure}[t]
	\centering
	\begin{tikzpicture}[scale=.75]
	
	\draw[white] (0,0) -- (3,0) -- (3,4) -- (0,4) -- cycle;
	\foreach \x in {.5,...,2.5}
	\foreach \y in {.5,...,3.5} {
		\node at (\x,\y) {$\bullet$};
	}
	
	\node[above] at (4.5,1.8) {\Large $\xhookrightarrow{\qquad}$};
	
	\begin{scope}[xshift=6cm]
	\draw (0,0) -- (3,0) -- (3,4) -- (0,4) -- cycle;
	\foreach \x in {.5,...,2.5}
	\foreach \y in {.5,...,3.5} {
		\node at (\x,\y) {$\circ$};
	}
	\end{scope}
	
	\node[above] at (11.5,1.8) {\Large $\xrightarrow{A_{c,\boldsymbol{\lambda},\boldsymbol{\tau}}}$};
	
	\begin{scope}[xshift=13cm,x={(1,.2)}, y={(.5,.75)}]
	\draw (0,0) -- (3,0) -- (3,4) -- (0,4) -- cycle;
	\foreach \x in {.5,...,2.5}
	\foreach \y in {.5,...,3.5} {
		\node at (\x,\y) {$\circ$};
	}
	\end{scope}
	
	\node at (1.5,4.5) {$U : \mathbb{Z}^N \rightarrow \mathbb{C}$};
	\node at (7.5,4.5) {$\widetilde{u} = u \circ A_{c,\boldsymbol{\lambda},\boldsymbol{\tau}} : \mathbb{R}^N \rightarrow \mathbb{C}$};
	\node at (16,4.5) {$u : \mathbb{R}^N \rightarrow \mathbb{C}$};
	
	\node at (1.5,-1) {$L(U(\square(\boldsymbol{n}))) $};
	\node at (3.8,-1) {$=$};
	\node at (7.5,-1) {$\displaystyle \widetilde{\cL}_{\mathrm{disc}}[\widetilde{u}(\boldsymbol{n})]$};
	\node at (11,-1) {$=$};
	\node at (15,-1) {$\displaystyle \mathcal{L}_{\mathrm{disc}}[u(A_{c,\boldsymbol{\lambda},\boldsymbol{\tau}}(\boldsymbol{n}))]$};

	\node at (1.5,-2.5) {$\displaystyle \sum L(U(\square)) $};
	\node at (3.8,-2.5) {$=$};
	\node at (7.5,-2.5) {$\displaystyle \int_\Gamma \cL_\mod[\widetilde{u}] \, \mathrm{d} t_{i_1} \wedge \ldots \wedge \mathrm{d} t_{i_d}$};
	\node at (11,-2.5) {$=$};
	\node at (15,-2.5) {$\displaystyle \int_{A_{c,\boldsymbol{\lambda},\boldsymbol{\tau}}(\Gamma)} \cL_\miwa[u] \, \eta_{i_1} \wedge \ldots \wedge \eta_{i_d}$};
	\end{tikzpicture}
	\caption{Visualization of the lattice, the straight embedding from Section \ref{sec-straight-limit}, and the skew embedding by the Miwa correspondence.}
	\label{fig-miwa-embedding}
\end{figure}

Let $\widetilde{u} = u \circ A_{c,\mathsmaller \blambda,\btau}$. The modified Lagrangian in Miwa coordinates is given by
\begin{align}
\cL_\miwa([u],\lambda_1,\ldots, \lambda_d) \Big|_{\bt}
&= \cL_\mod([\widetilde{u}],1)  \Big|_\Asub \notag\\
&= \sum_{i_1,\ldots,i_d=0}^\infty \frac{B_{i_1} \ldots B_{i_d}}{i_1! \ldots i_d!} \partial_1^{i_1} \ldots \partial_d^{i_d} \cL_\disc([u],\lambda_1,\ldots, \lambda_d)  \Big|_{\bt} , \label{Lmiwa}
\end{align}
where $\cL_\disc$ is given by Equation \eqref{Ldisc}. The relation between $U$, $\widetilde{u}$, and $u$ is illustrated in Figure \ref{fig-miwa-embedding}. Note that, although $\widetilde{u} \big|_\Asub = u \big|_\bt$, no such equality holds for derivatives, hence 
\begin{align*}
[\widetilde{u}] \Big|_\Asub \neq [u] \Big|_{\bt}. 
\end{align*}

\begin{lemma}\label{lemma-meshed-miwa}
Consider a filled-in plaquette of the embedded lattice, ${A_{c,\blambda,\btau}(\blacksquare_{i_1,\ldots,i_d}(\bn))}$, and let $\eta_k$ be the 1-forms dual to the Miwa shifts, 
\[ \eta_k = (A_{c,\blambda,\btau}^{-1})^* \d t_k, \]
where ${}^*$ denotes the pullback. Then as a formal power series, $\cL_\miwa$, defined in Equation \eqref{Lmiwa}, satisfies
\begin{equation}\label{Miwa-integral}
\int_{{A_{c,\blambda,\btau}(\blacksquare_{i_1,\ldots,i_d}(\bn))}} \cL_\miwa([u],\lambda_{i_1}, \ldots, \lambda_{i_d}) \, \eta_{i_1} \wedge \ldots \wedge \eta_{i_d} = L_\disc(\square_{i_1,\ldots,i_d}(\bn)) .
\end{equation}
For any truncation of the power series $\cL_\miwa$, Equation \eqref{Miwa-integral} holds with a defect of the corresponding order in $\lambda_{i_1}, \ldots, \lambda_{i_d}$.
\end{lemma}
\begin{proof}
In Equation \eqref{meshed} we have the corresponding result for $\cL_\mod$, so the proof is a simple change of variables:
\begin{align*}
L_\disc(U(\square_{i_1,\ldots,i_d}(\bn)))
&= \int_{\blacksquare_{i_1,\ldots,i_d}(\bn)} \cL_\mod([\widetilde{u}],1) \Big|_\bt \, \d t_{i_1} \wedge \ldots \wedge \d t_{i_d} \\ 
&= \int_{A_{c,\mathsmaller \blambda,\btau}(\blacksquare_{i_1,\ldots,i_d}(\bn))} \cL_\mod([\widetilde{u}],1) \Big|_\Asub \, (A_{c,\blambda,\btau}^{-1})^* (\d t_{i_1} \wedge \ldots \wedge \d t_{i_d}) \\
&= \int_{{A_{c,\mathsmaller \blambda,\btau}(\blacksquare_{i_1,\ldots,i_d}(\bn))}} \cL_\miwa([u],\lambda_{i_1}, \ldots, \lambda_{i_d}) \Big|_\bt \, \eta_{i_1} \wedge \ldots \wedge \eta_{i_d} . \qedhere
\end{align*}
\end{proof}
We want to use this result for plaquettes in arbitrary directions. This suggests the Lagrangian $d$-form
\[ \sum_{1 \leq i_1 < \ldots < i_d \leq N} \cL_\miwa([u],\lambda_{i_1},\ldots,\lambda_{i_d}) \, \eta_{i_1} \wedge \ldots \wedge \eta_{i_d}. \]
Up to a truncation error, this $d$-form can be written in a much more convenient way. Let $\cT_N$ denote truncation of a power series after degree $N$ in each variable,
\[ \cT_N \Bigg( \sum_{i_1,\ldots,i_d=1}^\infty \lambda_1^{i_1} \ldots \lambda_d^{i_d} \, f_{i_1,\ldots,i_d} \Bigg) 
= \sum_{i_1,\ldots,i_d=1}^N \lambda_1^{i_1} \ldots \lambda_d^{i_d} \, f_{i_1,\ldots,i_d}.
\]

\begin{lemma}\label{lemma-miwa-to-pluri}
Assume that every term in the power series $\cL_\miwa$ \eqref{Lmiwa} is of strictly positive degree in each $\lambda_i$,
\begin{equation}\label{Lmiwa-expansion}
\cL_\miwa([u],\lambda_1,\ldots, \lambda_d) = \sum_{i_1,\ldots,i_d=1}^\infty (-1)^{i_1+\ldots+i_d} c^d \frac{\lambda_1^{i_1}}{i_1} \ldots \frac{\lambda_d^{i_d}}{i_d} \cL_{i_1,\ldots,i_d}[u] ,
\end{equation}
then
\[
 \sum_{\substack{1 \leq j_1 < \ldots \\ < j_d \leq N}} 
\cT_N \big( \cL_\miwa([u],\lambda_{j_1},\ldots,\lambda_{j_d}) \big) \, \eta_{j_1} \wedge \ldots \wedge \eta_{j_d}
=  \sum_{\substack{1 \leq i_1 < \ldots \\ < i_d \leq N}}  \cL_{i_1,\ldots,i_d}[u] \, \d t_{i_1} \wedge \ldots \wedge \d t_{i_d} .
\]
\end{lemma}

Note in Equation \eqref{Lmiwa-expansion} that the factors $(-1)^{i_1+\ldots+i_d} c^d \frac{\lambda_1^{i_1}}{i_1} \ldots \frac{\lambda_d^{i_d}}{i_d}$ are terms of $(d \times d)$-minors of the transformation matrix $c M_{\blambda}$. 

\begin{proof}[Proof of Lemma \ref{lemma-miwa-to-pluri}]
First observe that, just like the discrete Lagrangian, the Lagrangian $\cL_\miwa([u],\lambda_{i_1},\ldots,\lambda_{i_d})$ is skew-symmetric as a function of $(\lambda_{i_1},\ldots,\lambda_{i_d})$. Therefore, the coefficients $\cL_{i_1,\ldots,i_d}[u]$ are skew-symmetric as a function of $(i_1,\ldots,i_d)$.

We pair the form
\begin{align*}
\cL &=\sum_{\substack{1 \leq i_1 < \ldots < i_d \leq N}}  \cL_{i_1,\ldots,i_d}[u] \, \d t_{i_1} \wedge \ldots \wedge \d t_{i_d} \\
&= \frac{1}{d!} \sum_{i_1,\ldots,i_d=1}^N \cL_{i_1,\ldots,i_d}[u] \, \d t_{i_1} \wedge \ldots \wedge \d t_{i_d}
\end{align*}
with a $d$-tuple of vectors $(\fv_{j_1}, \ldots, \fv_{j_d}) = (c M_{\blambda} \fe_{j_1}, \ldots, c M_{\blambda} \fe_{j_d})$:
\begin{align*}
\big\langle \cL , \left( \fv_{j_1}, \ldots, \fv_{j_d} \right) \big\rangle 
&= \frac{1}{d!} \sum_{i_1,\ldots,i_d=1}^N  \left( \cL_{i_1,\ldots,i_d}[u]
\sum_{\sigma \in S_d} \left( \sgn(\sigma) \prod_{k=1}^d \left\langle \d t_{i_{\sigma(k)}} , \fv_{j_k} \right\rangle \right) \right) .
\end{align*}
Due to the skew-symmetry of $\cL_{i_1,\ldots,i_d}[u]$, this can be written as
\[ \big\langle \cL , \left( \fv_{j_1}, \ldots, \fv_{j_d} \right) \big\rangle 
= \frac{1}{d!} \sum_{i_1,\ldots,i_d=1}^N \, \sum_{\sigma \in S_d} 
\left( \cL_{i_{\sigma(1)},\ldots,i_{\sigma(d)}}[u]  \prod_{k=1}^d \left\langle \d t_{i_{\sigma(k)}} , \fv_{j_k} \right\rangle \right) . \]
Since the first sum is over all $d$-tuples $(i_1,\ldots,i_d)$ with strictly positive integer entries, permuting $(i_1,\ldots,i_d)$ yields a different term of this sum. Hence the additional summation over permutations $\sigma \in S_d$ amounts to multiplication by $d!$. We find
\begin{align*}
\big\langle \cL , \left( \fv_{j_1}, \ldots, \fv_{j_d} \right) \big\rangle
&= \sum_{i_1 , \ldots , i_d = 1}^N  \cL_{i_1,\ldots,i_d}[u] \prod_{k=1}^d \left\langle \d t_{i_k} , \fv_{j_k} \right\rangle \\
&=  \sum_{i_1 , \ldots , i_d = 1}^N \cL_{i_1,\ldots,i_d}[u] \prod_{k=1}^d (-1)^{i_k} c \frac{\lambda_{j_k}^{i_k}}{i_k} \\
&= \cT_N \big( \cL_\miwa([u],\lambda_{j_1},\ldots,\lambda_{j_d}) \big) . \qedhere
\end{align*}
\end{proof}

The higher the dimension $N$ we consider, the more accurately discrete and continuous critical fields will correspond to each other. This accuracy is made precise with the following notion.
 
\begin{defi}
	Let $\Gamma$ be a finite discrete surface in $\Z^N$. A discrete field $U: \Z^N \mapsto \C$ is \emph{$k$-critical} for the action
	\[ S = \sum_{\square_{i_1,\ldots,i_d}(\bn) \in \Gamma} L \big( U(\square_{i_1,\ldots,i_d}(\bn)) ,\lambda_{i_1}, \ldots, \lambda_{i_d} \big) . \]
	if for any $\bn \in \Z^N$ there holds
	\[ \der{S}{U(\bn)} = \O \big( \big( \lambda_1^{k+1} + \ldots + \lambda_N^{k+1} \big) |S| \big). \]
\end{defi}

We now arrive at our main result.

\begin{thm}\label{thm}
Let $L_\disc$ be a discrete Lagrangian $d$-form, such that every term in the corresponding power series $\cL_\miwa$ \eqref{Lmiwa} is of strictly positive degree in each $\lambda_i$, i.e.\@ such that $\cL_\miwa$ is of the form \eqref{Lmiwa-expansion}. Consider the differential $d$-form
\[ \cL  = \sum_{1 \leq i_1 < \ldots < i_d \leq N}
\cL_{i_1,\ldots,i_d}[u] \, \d t_{i_1} \wedge \ldots \wedge \d t_{i_d}, \]
built out of the coefficients of $\cL_\miwa$. Then a field $u:\R^N \rightarrow \C$ is a solution to the continuous pluri-Lagrangian problem for $\cL$ if and only if the corresponding discrete fields 
\[ U_{\btau}:\Z^N \rightarrow \C: \bn \mapsto u(A_{c,\blambda,\btau}(\bn)), \qquad \btau \in \R^N,\]
with $A_{c,\blambda,\btau}$ given by equation \eqref{miwaA}, are $N$-critical for the discrete pluri-Lagrangian problem for $L_\disc$.
\end{thm}

\begin{proof}
	Consider a bounded $d$-surface $\Gamma$ in $\R^N$ that does not depend on $\blambda$. We can approximate it by an image of a discrete surface $\overline{\Gamma}$ under the Miwa embedding, with an error of order $\O( \lambda_1 + \ldots + \lambda_N)$,
	\begin{equation}\label{surface-approx}
	\int_{\Gamma} \cL[u] =  \int_{A_{c,\mathsmaller \blambda,\btau} \left( \overline{\Gamma} \right)} \cL[u] + \O( \lambda_1 + \ldots + \lambda_N).
	\end{equation}
	This idea of approximating any given surface by a \emph{stepped surface} was used in \cite{suris2016lagrangian} to derive the multi-time Euler-Lagrange equations.
	
	By Lemmas \ref{lemma-meshed-miwa} and \ref{lemma-miwa-to-pluri} it follows that for any continuous field $u$ and its corresponding discrete field $U_{\btau}$ we have
	\begin{equation}\label{disc-approx}
	\int_{A_{c,\mathsmaller \blambda,\btau} \left( \overline{\Gamma} \right)} \cL[u] = \sum_{\square_{i_1,\ldots,i_d}(\bn) \in \overline{\Gamma}} L_\disc(U_{\btau}(\square_{i_1,\ldots,i_d}(\bn)),\lambda_{i_1}, \ldots, \lambda_{i_d}) + \O \big( \lambda_1^{N+1} + \ldots + \lambda_N^{N+1} \big),
	\end{equation}
	hence if the continuous field $u$ is critical, then the discrete field $U_{\btau}$ is $N$-critical.
	
	From Equations \eqref{surface-approx} and \eqref{disc-approx} it follows that
	\[  \int_{\Gamma} \cL[u] = L_\disc(U_{\btau}(\square_{i_1,\ldots,i_d}(\bn)),\lambda_{i_1}, \ldots, \lambda_{i_d}) + \O( \lambda_1 + \ldots + \lambda_N) . \]
	Now assume that the discrete field is $0$-critical. Then for any $\blambda$-independent variation of $u$ that is zero near the boundary of $\Gamma$, we have that 
	\[ \delta \int_{\Gamma} \cL[u] =  \O( \lambda_1 + \ldots + \lambda_N) . \]
	Since the left hand side is independent of $\blambda$, it must be exactly zero. Hence $u$ is a critical field.
\end{proof}

Note that we did not just prove that discrete $N$-criticality is equivalent to continuous criticality, but also that discrete $0$-criticality \emph{implies} continuous criticality. Hence if a discrete field, obtained from a $\blambda$-independent continuous field $u$ by the relation $U_{\btau}(\bn) = u(A_{c,\blambda,\btau}(\bn))$, is just $0$-critical, then it is automatically $N$-critical. Of course this does not hold for arbitrary discrete fields.

\subsection{Eliminating alien derivatives}

Unlike in the classical Lagrangian framework, Euler-Lagrange equations in the pluri-Lagrangian context are often \emph{evolutionary}. For a pluri-Lagrangian $d$-form in $\R^N$, this means the hierarchy is of the form
\[ u_{t_k} = f_k[u] \qquad \text{for } k \in \{d,d+1,\ldots,N\} , \]
 where the $f_k$ only depend on derivatives with respect to $t_1, \ldots, t_{d-1}$. Then the differential consequences of the multi-time Euler-Lagrange equations can be written in a similar form,
\begin{equation}\label{fullel}
	u_I = f_I[u] \qquad \text{with $I \ni t_k$ for some } k \in \{d,d+1,\ldots,N\},
\end{equation}
where the $I$ in $f_I$ is a label, not a partial derivative. In this context it is natural to consider the first $d-1$ coordinates $t_1, \ldots,t_{d-1}$ as space coordinates and the others as time coordinates. If the multi-time Euler-Lagrange equations are not evolutionary, Equation \eqref{fullel} still holds for a smaller set of multi-indices $I$.

\begin{defi}
A function $f[u]$ is called
\begin{enumerate}[$(a)$]
\item \emph{$\{i_1,\ldots,i_{d}\}$-native} if it only depends on $u$ and its derivatives with respect to $t_{i_1},\ldots,$ $t_{i_d}$ and with respect to the space coordinates $t_1,\ldots,t_{d-1}$,
\item \emph{$\{i_1,\ldots,i_{d}\}$-alien} if it is not $\{i_1,\ldots,i_{d}\}$-native, i.e.\@ if it depends on a $t_k$-derivative with $k \not\in \{1,\ldots,d-1,i_1,\ldots,i_d\}$.
\end{enumerate}
A multi-index $I$ is said to be native or alien if the corresponding derivative $u_I$ is of that type.
\end{defi}

We would like the coefficient $\cL_{i_1,\ldots,i_d}$ to be $\{i_1,\ldots,i_d\}$-native. A naive approach would be to use the multi-time Euler-Lagrange equations \eqref{fullel} to eliminate all alien derivatives. Let $R_{i_1,\ldots,i_d}$ denote the operator that replaces all $\{i_1,\ldots,i_d\}$-alien derivatives for which \eqref{fullel} provides an expression. If the hierarchy is evolutionary $R_{i_1,\ldots,i_d}$ replaces all alien derivatives. We denote the resulting pluri-Lagrangian coefficients by
\[\overline{\cL}_{i_1,\ldots,i_d} = R_{i_1,\ldots,i_d}(\cL_{i_1,\ldots,i_d})\]
and the $d$-form with these coefficients by $\overline{\cL}$. A priori there is no reason to believe that the $d$-form $\overline{\cL}$ will be equivalent to the original pluri-Lagrangian $d$-form $\cL$. For example, the 2-dimensional Lagrangian 1-form $\cL(u,u_{t_1},u_{t_2}) = u_{t_2}^2 \d t_1 + u_{t_1}^2 \d t_2$ leads to the multi-time Euler-Lagrange equations $u_{t_1} = u_{t_2} =0$, but any field $u$ is critical for the Lagrangian $\overline{\cL}(u,u_{t_1},u_{t_2}) = 0$. However, in many cases the pluri-Lagrangian structure guarantees that $\cL$ and $\overline{\cL}$ have the same critical fields.

\begin{thm}\label{thm-alien}
Assume that $R_{i_1,\ldots,i_d}$ commutes with the operators $\D_{t_{i_1}}, \ldots, \D_{t_{i_d}}$. If either
\begin{itemize}
\item $d = 1$ and $\cL_1[u]$ does not depend on any alien derivatives, or
\item $d = 2$ and for all $j$ the coefficient $\cL_{1j}[u]$ does not contain any alien derivatives,
\end{itemize}
then every critical field $u$ for the pluri-Lagrangian $d$-form $\cL$ is also critical for $\overline{\cL}$.
\end{thm}

In particular, the commutativity condition holds if the equations in the hierarchy are evolutionary, or more generally, if none of their left hand sides are a mixed partial derivative. The condition for $d=2$ might seem restrictive, but given a Lagrangian $2$-form, we can often find an equivalent one with native coefficients $\cL_{1j}[u]$ by inspection.

\begin{proof}[Proof of Theorem \ref{thm-alien}]
In this proof we consider the variation operator $\delta$ as the vertical exterior derivative in the variational bicomplex. A short introduction to the variational bicomplex is given in Appendix \ref{appendix}. 

First we consider the case $d = 1$. Let 
\[ F_{i,J}[u] = R_i(u_J),\]
i.e.\@ $F_{i,J} = u_J$ if $J$ is $\{i\}$-native and $F_{i,J}$ is the native replacement for $u_J$ otherwise. Note that $\D_{t_i}$ and $R_i$ commute, hence $\D_{t_i} F_{i,J} = F_{i,Jt_i}$.  We have
\begin{align*}
\delta \overline{\cL}
&= \sum_{1 \leq i \leq N} \sum_{J} R_i\!\left( \der{\cL_i}{u_J} \right) \delta F_{i,J} \wedge \d t_i \\
&= \sum_{1 \leq i \leq N} 
\sum_J R_i\!\left( \var{i}{\cL_i}{u_J} + \D_{t_i} \var{i}{\cL_i}{u_{J t_i}} \right) \delta F_{i,J}
\wedge \d t_i \\
&= \sum_{1 \leq i \leq N} 
 \left( \sum_{J \not\ni t_i} R_i\!\left( \var{i}{\cL_i}{u_J} \right) \delta F_{i,J} + \sum_{J}
\D_{t_i} \!\left( R_i\!\left( \var{i}{\cL_i}{u_{J t_i}} \right) \delta F_{i,J} \right) \right) \wedge \d t_i  .
\end{align*}
Hence on solutions of the pluri-Lagrangian problem for $\cL$ there holds that
\[ \delta \overline{\cL}
= \sum_{1 \leq i \leq N} \left( \D_{t_i} \sum_J \var{1}{\cL_1}{u_{J t_1}} \delta F_{i,J} \right) \wedge \d t_i .\]
Using the assumption that no alien derivatives occur in $\cL_1$, we can simplify this to
\[ \delta \overline{\cL} 
= \sum_{1 \leq i \leq N} \D_{t_i}\! \left( \sum_{\alpha = 0}^\infty \der{\cL_1}{u_{t_1^{\alpha+1}}} \delta u_{t_1^\alpha} \right) \wedge \d t_i 
= \d\! \left( - \sum_{\alpha = 0}^\infty \der{\cL_1}{u_{t_1^{\alpha+1}}} \delta u_{t_1^\alpha} \right) . \]
The fact that $\delta \overline{\cL}$ is exact with respect to $\d$ implies that $\delta \int_\Gamma \overline{\cL} = 0$ for all curves $\Gamma$ and all variations that are zero on the endpoints of $\Gamma$. Hence $u$ is a solution to the pluri-Lagrangian problem for $\overline{\cL}$.

Now we consider the case $d = 2$. Let \[F_{ij,J} = R_{ij}(u_J).\] Note that $R_{ij}$ commutes with both $\D_{t_i}$ and $\D_{t_j}$. We have
\begin{align*}
&\delta \overline{\cL} = \sum_{1 \leq i < j  \leq N} \sum_J 
R_{ij}\!\left( \der{\cL_{ij}}{u_J} \right) \delta F_{ij,J} \wedge \d t_i \wedge \d t_j
\\
&= \sum_{1 \leq i < j \leq N} \sum_J 
 R_{ij}\!\left( \var{ij}{\cL_{ij}}{u_J} + \D_{t_i} \var{ij}{\cL_{ij}}{u_{Jt_i}} + \D_{t_j} \var{ij}{\cL_{ij}}{u_{Jt_j}} + \D_{t_i} \D_{t_j} \var{ij}{\cL_{ij}}{u_{Jt_it_j}} \right) \delta F_{ij,J} 
  \wedge \d t_i \wedge \d t_j
\\
&= \sum_{1 \leq i < j \leq N} \Bigg( 
\sum_{J \not\ni t_i,t_j}  R_{ij}\!\left( \var{ij}{\cL_{ij}}{u_J} \right) \delta F_{ij,J} 
+ \sum_{J \not\ni t_j}  \D_{t_i} \!\left( R_{ij}\!\left( \var{ij}{\cL_{ij}}{u_{Jt_i}} \right) \delta F_{ij,J} \right) \\
&\qquad + \sum_{J \not\ni t_i} \D_{t_j} \!\left( R_{ij}\!\left( \var{ij}{\cL_{ij}}{u_{Jt_j}} \right) \delta F_{ij,J} \right) + \sum_{J} \D_{t_i} \D_{t_j}  \!\left( R_{ij}\!\left( \var{ij}{\cL_{ij}}{u_{Jt_it_j}} \right) \delta F_{ij,J} \right) \Bigg)
 \wedge \d t_i \wedge \d t_j .
\end{align*}
On solutions of the pluri-Lagrangian problem for $\cL$ there holds that
\begin{align*}
\delta \overline{\cL}
&= \sum_{1 \leq i < j \leq N} \Bigg( 
\sum_{J \not\ni t_j} \D_{t_i} \!\left( \var{1j}{\cL_{1j}}{u_{Jt_1}} \delta F_{ij,J} \right) - \sum_{J \not\ni t_i}\D_{t_j} \!\left( \var{1i}{\cL_{1i}}{u_{Jt_1}} \delta F_{ij,J} \right) \\
&\hspace{37mm} + \sum_{J}  \D_{t_i} \D_{t_j} \!\left( \left( \var{1j}{\cL_{1j}}{u_{Jt_1t_j}} - \var{1j}{\cL_{1i}}{u_{Jt_1t_i}} \right) \delta F_{ij,J} \right)
\Bigg) \wedge \d t_i \wedge \d t_j ,
\end{align*}
where we have left out the $R_{ij}$ because the $\cL_{1j}$ do not contain any alien derivatives. For the same reason, only terms where $J$ is $\{i,j\}$-native can be nonzero, so in all nonvanishing terms we find $F_{ij,J} = u_J$. Therefore,
\begin{align*}
\delta \overline{\cL}
&= \sum_{1 \leq i < j \leq N} \left( 
\D_{t_i} \!\Bigg( 
\sum_{J \not\ni t_j} \var{1j}{\cL_{1j}}{u_{Jt_1}} \delta u_J +
\sum_J \D_{t_j} \!\left(  \var{1j}{\cL_{1j}}{u_{Jt_1t_j}}  \delta u_J \Bigg) 
\right) \right.
\\
&\hspace{22mm}\left. -\D_{t_j} \!\Bigg( 
\sum_{J \not\ni t_i} \var{1i}{\cL_{1i}}{u_{Jt_1}} \delta u_J +
\sum_J \D_{t_j} \!\left( \var{1i}{\cL_{1i}}{u_{Jt_1t_i}}  \delta u_J \right) 
\Bigg) 
\right) \wedge \d t_i \wedge \d t_j
\\
&= \d \!\left( \sum_{1 \leq j \leq N} 
\left( 
\sum_{J \not\ni t_j} \var{1j}{\cL_{1j}}{u_{Jt_1}} \delta u_J +
\sum_J \D_{t_j} \!\left(  \var{1j}{\cL_{1j}}{u_{Jt_1t_j}}  \delta u_J \right) 
\right) \wedge \d t_j \right) .
\end{align*}
This implies that $\delta \int_\Gamma \overline{\cL} = 0$ for all surfaces $\Gamma$ and all variations that are zero on the boundary of $\Gamma$. Hence $u$ is a solution to the pluri-Lagrangian problem for $\overline{\cL}$.
\end{proof}

\section{Examples}\label{sec-examples}

The plan for this section is as follows. We begin with the 1-form case and discuss the continuum limit of the discrete Toda lattice. After that we present three examples for the 2-form case. The first one is a linear quad equation. This will help us understand how to proceed for the two non-linear quad equations that follow, H1 and Q1$_{\delta = 0}$ from the ABS list. In each of the examples we first perform the continuum limit on the level of equations and then discuss the pluri-Lagrangian structure.

\subsection{Toda lattice}

\subsubsection{Equation}

The Toda lattice consists of a number of particles on a line with an exponential nearest-neighbor force. If we denote the displacements of the particles from their equilibrium positions by 
\[ q(t) = (\q{0}(t),\q{1}(t),\ldots,\q{N}(t)), \]
then their motion is described by the equation
\[ \frac{\d^2 \q{k}}{\d t^2} = \exp\!\left( \q{k+1} - \q{k} \right) - \exp\!\left( \q{k} - \q{k-1} \right) . \]
There are two common conventions regarding boundary conditions: periodic (formally $\q{N+1} \equiv \q{1}$) and open-end (formally $\q{0} \equiv +\infty$ and $\q{N+1} \equiv -\infty$). An integrable discretization of the Toda lattice is given by \cite[Chapter 5]{suris2003problem}
\begin{equation}\label{dtoda}
\begin{split}
& \frac{1}{\lambda_i} \left( \exp\!\left( \Q{k}_i - \Q{k} \right)-  \exp\!\left( \Q{k} - \Q{k}_{-i} \right) \right) \\
&\qquad  + \lambda_i \left(  \exp\!\left( \Q{k} - \Q{k-1}_i \right) -  \exp\!\left( \Q{k+1}_{-i} - \Q{k} \right) \right) =  0,
\end{split}
\end{equation}
where the subscripts $i$ and $-i$ denote forward and backward shifts respectively and $\lambda_i$ is a lattice parameter.

We use the Miwa correspondence \eqref{miwa} with $c = 1$ to identify discrete steps with continuous time shifts
\begin{align*}
\Q{k} &= \q{k}(t_1,t_2,t_3,\ldots) , \\
\Q{k}_i &= \q{k}\!\left(t_1 + \lambda_i,t_2 - \frac{\lambda_i^2}{2},t_3 + \frac{\lambda_i^3}{3},\ldots \right) , \\
\Q{k}_{-i} &= \q{k}\!\left(t_1 - \lambda_i,t_2 + \frac{\lambda_i^2}{2},t_3 - \frac{\lambda_i^3}{3},\ldots \right) .
\end{align*}
We plug these identifications into Equation \eqref{dtoda} and perform a Taylor expansion in $\lambda_i$:
\begin{align*}
&\left( -\exp\!\left( \q{k+1} - \q{k} \right) + \exp\!\left( \q{k} - \q{k-1} \right) + \q{k}_{11} \right) \lambda_i \\
&\qquad + \left(\exp\!\left( \q{k+1} - \q{k} \right) \q{k+1}_1 - \exp\!\left( \q{k} - \q{k-1} \right) \q{k-1}_1 + \q{k}_1 \q{k}_{11} - \q{k}_{12} \right) \lambda_i^{2} = \O(\lambda_i^3),
\end{align*}
where the subscripts $1$ and $2$ of $q$ are a shorthand for $t_1$ and $t_2$, denoting partial derivatives. As long as one remembers that discrete fields are printed in upper case and continuous fields in lower case, there should be no confusion between partial derivatives and lattice shifts.
In the leading order term we recognize the first Toda equation 
\begin{equation}\label{toda1}
\q{k}_{11} = \exp\!\left( \q{k+1} - \q{k} \right) - \exp\!\left( \q{k} - \q{k-1} \right) .
\end{equation}
Using this equation, we find that the coefficient of $\lambda_i^2$ is
\begin{align*}
& \exp\!\left( \q{k+1} - \q{k} \right) \q{k+1}_1 - \exp\!\left( \q{k} - \q{k-1} \right) \q{k-1}_1 + \q{k}_1 \q{k}_{11} - \q{k}_{12} \\
&= \exp\!\left( \q{k+1} - \q{k} \right) \left( \q{k+1}_1 - \q{k}_1 \right)  - \exp\!\left( \q{k} - \q{k-1} \right) \left( \q{k-1}_1 - \q{k}_1 \right) + 2 \q{k}_1 \q{k}_{11} - \q{k}_{12} \\
&= \D_{t_1} \!\left( \exp\!\left( \q{k+1} - \q{k} \right) + \exp\!\left( \q{k} - \q{k-1} \right) + \left( \q{k}_1 \right)^2 - \q{k}_2 \right).
\end{align*}
Under the differentiation one can recognize the second Toda equation
\begin{equation}\label{toda2}
\q{k}_2 = \left( \q{k}_1 \right)^2 + \exp\!\left( \q{k+1} - \q{k} \right) + \exp\!\left( \q{k} - \q{k-1} \right) .
\end{equation}
Similarly, the higher order terms correspond to the subsequent equations of the Toda hierarchy.

\subsubsection{Pluri-Lagrangian structure}

A pluri-Lagrangian structure for the discrete Toda equation was studied in \cite{boll2013multi}. The Lagrangian is given by
\begin{equation}\label{dtoda-lag}
\begin{split}
L \big( Q,Q_i,\lambda_i \big) &= \frac{1}{\lambda_i} \sum_k \left( \exp\!\left( \Q{k}_i - \Q{k} \right) - 1 - \left( \Q{k}_i - \Q{k} \right) \right) \\
&\quad - \lambda_i \sum_k \exp\!\left( \Q{k} - \Q{k-1}_i \right) .
\end{split}
\end{equation}
Performing a Taylor expansion and applying the Euler-Maclaurin formula as in Section \ref{sec-pluri-limit}, we obtain
\[\cL_\miwa([q],\lambda_i) = \sum_{j=1}^\infty (-1)^{j+1} \frac{\lambda_i^j}{j} \cL_j[q] \]
with coefficients
\begin{align*}
\cL_1 &= \sum_k \left( \frac{1}{2} \left( \q{k}_1 \right)^2 - \exp\!\left( \q{k} - \q{k-1} \right) \right) , \\
\cL_2 &= \sum_k \left( \q{k}_1 \q{k}_2 - \frac{1}{3}\left( \q{k}_1 \right)^3 - \left( \q{k}_1 + \q{k-1}_1 \right) \exp\!\left( \q{k} - \q{k-1} \right) \right) , \\
\cL_3 &= \sum_k \bigg( 
- \frac{1}{4} \left(  \left( \q{k+1}_1 \right)^2 + 4 \q{k+1}_1 \q{k}_1 + \left( \q{k}_1 \right)^{2} + \q{k+1}_{11} \right) \exp\!\left( \q{k+1} - \q{k} \right) \\
&\hspace{1.5cm}
+ \frac{1}{4} \left( - \q{k+1}_{11} + \q{k}_{11} - 3 \q{k}_{2} - 3 \q{k+1}_2 \right) \exp\!\left( \q{k+1} - \q{k} \right) \\
&\hspace{1.5cm}
 +  \frac{1}{8} \left( \q{k}_1 \right)^4 - \frac{3}{4} \left( \q{k}_1 \right)^2 \q{k}_2 - \frac{1}{8} \left( \q{k}_{11} \right)^2  + \frac{3}{8} \left( \q{k}_2 \right)^2 + \q{k}_1 \q{k}_3 \bigg) , \\
&\ \vdots
\end{align*}
By Theorem \ref{thm}, these are the coefficients of a pluri-Lagrangian 1-form $\cL = \sum_i \cL_i \,\d t_i$ for the Toda hierarchy \eqref{toda1}, \eqref{toda2}, $\cdots$.

Note that $\cL_2$ contains derivatives with respect to $t_1$, which are alien. However, there is no equation in the hierarchy that can be used to eliminate these derivatives. In this example we have to tolerate the alien derivative $\q{k}_1$. The next coefficient, $\cL_3$, contains second derivatives with respect to $t_1$ and derivatives with respect to $t_2$. We replace these using the first and second Toda equation and find
\begin{align*}
	\overline{\cL}_3 &= \sum_k \bigg( -\frac{1}{4} \left( \q{k}_1 \right)^{4} - \left( \left( \q{k+1}_1 \right)^2 + \q{k+1}_1 \q{k}_1 + \left( \q{k}_1 \right)^2 \right) \exp\!\left( \q{k+1} - \q{k} \right) \\
	&\hspace{1.5cm} + \q{k}_1 \q{k}_3 - \exp\!\left( \q{k+2} - \q{k} \right) - \frac{1}{2} \exp\!\left( 2 (\q{k+1} - \q{k}) \right) \bigg) .
\end{align*}
Similarly one can obtain $\overline{\cL_i}$ for $i \geq 4$. By Theorem \ref{thm-alien}, the corresponding $1$-form $\overline{\cL}$ is equivalent to $\cL$. The Lagrangian 1-form $\overline{\cL}$ is identical to the one that was found in \cite{petrera2017variational} using the variational symmetries of the Toda lattice.

\subsection{A linear quad equation}
\label{sec-Qlin} 

\subsubsection{Equation}

Consider the linear quad equation
\begin{equation}
\label{Q-lin} (\alpha_1 - \alpha_2)(U - U_{12}) = (\alpha_1 + \alpha_2)(U_1 - U_2) .
\end{equation}
It is a discrete analogue of the Cauchy-Riemann equations \cite{bobenko2015discrete} and also the linearization of the lattice potential KdV equation, which will be discussed in Section \ref{sec-H1}. Therefore all the results in this section are consequences of those in Section \ref{sec-H1}. Nevertheless, this simple quad equation is a good subject to illustrate some of the subtleties of the continuum limit procedure.

To get meaningful equations in the continuum limit, we need to write the quad equation in a suitable form. Since in the Miwa correspondence the parameter enters linearly in the $t_1$-coordinate and with higher powers in the other coordinates, the leading order of the expansion of the shifts of $U$ will only contain derivatives with respect to $t_1$. Other derivatives enter at higher orders. Since we want to obtain PDEs in the continuum limit, not ODEs, we must require that the leading order of the expansion yields a trivial equation.

Written in terms of difference quotients, Equation \eqref{Q-lin} reads
\[
\frac{U_1 - U_2}{\alpha_1 - \alpha_2} = \frac{U - U_{12}}{\alpha_1 + \alpha_2} ,
\]
but setting $U=u(t_1,\ldots)$, $U_i = u(t_1+\alpha_i,\ldots)$, etc., this would yield $u_{t_1} = -u_{t_1}$ in the leading order of the expansion. In order to avoid this, we introduce new parameters $\lambda_i = \alpha_i^{-1}$. Then Equation \eqref{Q-lin} reads
\begin{equation}
\label{Q-lin-bis} 
\left( \frac{1}{\lambda_1} - \frac{1}{\lambda_2} \right)(U - U_{12}) - \left( \frac{1}{\lambda_1} + \frac{1}{\lambda_2} \right)(U_1 - U_2) = 0.
\end{equation}
or, equivalently,
\[ \frac{\lambda_1^2 - \lambda_2^2}{\lambda_1 \lambda_2} \left( \frac{U_1 - U_2}{\lambda_1 - \lambda_2} - \frac{U_{12}-U}{\lambda_1 + \lambda_2} \right) = 0. \]
Inside the brackets we find $u_{t_1} = u_{t_1}$ in the leading order if we set $U = u(t_1,\ldots)$, $U_i = u(t_1+ \lambda_i,\ldots)$, etc., which is trivial as desired.

We use the Miwa correspondence \eqref{miwa} with $c = -2$. This choice will give us a nice normalization of the resulting differential equations. We apply the Miwa correspondence to Equation \eqref{Q-lin-bis} and expand to find a double power series in $\lambda_1$ and $\lambda_2$,
\[ \sum_{i,j} \frac{4 (-1)^{i+j}}{ij} \cF_{ij}[u] \lambda_1^i \lambda_2^j = 0 , \]
where $\cF_{ji} = -\cF_{ij}$. The factor $(-1)^{i+j} \frac{4}{ij}$ is chosen to normalize the $\cF_{0j}$, but does not influence the final result. The first few of these coefficients are
\begin{align*}
\cF_{01} &= u_{t_2}, \\
\cF_{02} &= -u_{t_1t_1t_1} + \frac{3}{2}  u_{t_1t_2} + u_{t_3}, \\
\cF_{03} &= -\frac{4}{3}  u_{t_1t_1t_1t_1} + \frac{4}{3}  u_{t_1t_3} + u_{t_2t_2} + u_{t_4}, \\
\cF_{04} &=  - u_{t_1t_1t_1t_1t_1} - \frac{5}{3} u_{t_1t_1t_1t_2} + \frac{5}{4}  u_{t_1t_2t_2} + \frac{5}{4}  u_{t_1t_4} + \frac{5}{3} u_{t_2t_3} + u_{t_5},  \\
&\ \vdots 
\end{align*}
We see that the flows corresponding to even times are trivial. In the odd orders we find a hierarchy of linear equations,
\begin{align*}
u_{t_2} = 0, \qquad
u_{t_3} = u_{t_1t_1t_1}, \qquad
u_{t_4} = 0, \qquad
u_{t_5} = u_{t_1t_1t_1t_1t_1}, \qquad \cdots .
\end{align*}
For $i \geq 1$, the equations $\cF_{ij} = 0$ are consequences of these equations.

\subsubsection{Pluri-Lagrangian structure}

The linear quad equation \eqref{Q-lin} possesses a pluri-Lagrangian structure \cite{bobenko2015discrete,king2017quantum},
\begin{equation}
\label{linlag}
L(U,U_i,U_j,U_{ij},\alpha_i,\alpha_j) = U (U_i - U_i) - \frac{1}{2} \frac{\alpha_i + \alpha_j}{\alpha_i - \alpha_j} (U_i - U_j)^2 .
\end{equation}
The following Lemma will help us put this Lagrangian in a more convenient form.

\begin{lemma}\label{lemma-nulllag1}
$L_0(U,U_i,U_j,U_{ij},\alpha_i,\alpha_j) = (U + U_{ij})(U_i - U_j)$ is a null Lagrangian (i.e.\@ its multi-time Euler-Lagrange equations are trivially satisfied)
\end{lemma}
\begin{proof}
Consider the discrete 1-form given by $\eta(U,U_i) = U U_i$ and $\eta(U_i,U) = -U U_i$. Its discrete exterior derivative is
\[ \Delta \eta(U,U_i,U_{ij},U_j) = UU_i + U_iU_{ij} - U_{ij}U_j - U_j U = L_0.\]
Just like in the continuous case, this means that the action of $L_0$ over any discrete surface only depends on values of $U$ at the boundary of the surface. Hence all fields are critical with respect to variations in the interior.
\end{proof}
Using Lemma \ref{lemma-nulllag1}, we see that the Lagrangian \eqref{linlag} is equivalent to (denoted with $=$ by abuse of notation)
\[ L(U,U_i,U_j,U_{ij},\alpha_i,\alpha_j) = \frac{1}{2}(U_i - U_j) (U - U_{ij}) - \frac{1}{2} \frac{\alpha_i + \alpha_j}{\alpha_i - \alpha_j} (U_i - U_j)^2 , \]
or, in terms of the parameters $\lambda_k$,
\[ L(U,U_i,U_j,U_{ij},\lambda_i,\lambda_j) = \frac{1}{2}(U_i - U_j) (U - U_{ij}) + \frac{1}{2} \frac{\lambda_i + \lambda_j}{\lambda_i -\lambda_j} (U_i - U_j)^2 . \]
Since the Taylor expansion of $(U_i - U_j)^2$ contains a factor $\lambda_i - \lambda_j$, the expansion of the Lagrangian does not contain any negative order terms. In fact all zeroth order terms vanish as well, so Theorem \ref{thm} applies: the coefficients of the power series
\[\cL_\miwa([u],\lambda_1, \lambda_2) = \sum_{i,j=1}^\infty \frac{4(-1)^{i+j}}{ij} \cL_{ij}[u] \lambda_1^i \lambda_2^j \]
define a pluri-Lagrangian 2-form
\[ \cL = \sum_{1 \leq i < j \leq N} \cL_{ij} \, \d t_i \wedge \d t_j . \]

We find
\begin{align*}
\cL_{12} &= u_{t_1} u_{t_2} , \\
\cL_{13} &= -u_{t_1} u_{t_1t_1t_1} + \frac{3}{4} u_{t_2}^{2} + u_{t_1} u_{t_3} , \\
\cL_{23} &= -u_{t_1} u_{t_1t_1t_2} + u_{t_1t_1} u_{t_1t_2} - 2 u_{t_1t_1t_1} u_{t_2} - 3 u_{t_1t_2} u_{t_2} - 3 u_{t_1} u_{t_2t_2} + u_{t_2} u_{t_3} , \\
&\ \vdots
\end{align*}
We will not study this example in more detail. Instead we move on to one of its non-linear cousins.

\subsection{Lattice potential KdV (H1)}
\label{sec-H1}

The computations below were performed in the SageMath software system \cite{sagemath}. The code is available at \url{https://github.com/mvermeeren/pluri-lagrangian-clim}.

\subsubsection{Equation}

Consider equation H1 from the ABS list \cite{adler2003classification}, also known as the lattice potential Korteweg-de Vries (lpKdV) equation,
\begin{equation}\label{H1}
(V_{12} - V)(V_2 - V_1) = \alpha_1 - \alpha_2 .
\end{equation}
We would like write Equation \eqref{H1} in terms of difference quotients. To achieve this, we identify $\alpha_1 = -\lambda_1^{-2}$ and $\alpha_2 = -\lambda_2^{-2}$. Then Equation \eqref{H1} is equivalent to
\[ \frac{V_{12} - V}{\lambda_1 + \lambda_2} \frac{V_2 - V_1}{\lambda_2 - \lambda_1} = \frac{1}{\lambda_1^2 \lambda_2^2} . \]
The left hand side is now a product of meaningful difference quotients, but the right hand side explodes as the parameters tend to zero. (Setting $\alpha_i = -\lambda_i^2$ instead would cause the same problem as in the first attempt of Section \ref{sec-Qlin}.) To avoid this we make a non-autonomous change of variables
\[ V(n_1,\ldots,n_N) = U(n_1,\ldots,n_N) + \frac{n_1 }{\lambda_1} + \ldots  \frac{n_N}{\lambda_N} . \]
Then the lpKdV equation takes the form
\begin{equation}\label{lpKdV}
\left( \frac{1}{\lambda_1} + \frac{1}{\lambda_2} + U_{12} - U \right)\left( \frac{1}{\lambda_2} - \frac{1}{\lambda_1}  + U_2 - U_1 \right) = \frac{1}{\lambda_2^2} - \frac{1}{\lambda_1^2} .
\end{equation}
This is the form in which the lpKdV equation was originally found and studied, usually with parameters $p = \lambda_1^{-1}$ and $q = \lambda_2^{-1}$, see \cite{nijhoff1995discrete} for an overview. 
In terms of difference quotients, the equation reads 
\[ \frac{U_{12} - U}{\lambda_1 + \lambda_2} - \frac{U_2 - U_1}{\lambda_2 - \lambda_1} - 
\lambda_1 \lambda_2 \frac{U_{12} - U}{\lambda_1 + \lambda_2} \frac{U_2 - U_1}{\lambda_2 - \lambda_1} = 0 . \]
If we identify $U = u(t_1,\ldots)$, $U_i = u(t_1 + \lambda_i, \ldots)$, etc., then the negative powers of the parameters cancel. In the leading we find the tautological equation $u_{t_1} - u_{t_1} = 0$. Therefore, this form of the difference equation is a suitable candidate for the continuum limit.

Again we use the Miwa correspondence \eqref{miwa} with $c = -2$. From Equation \eqref{lpKdV} we find a double power series in $\lambda_1$ and $\lambda_2$,
\[ \sum_{i,j} \frac{4 (-1)^{i+j}}{ij} \cF_{ij}[u] \lambda_1^i \lambda_2^j = 0 , \]
where $\cF_{ji} = -\cF_{ij}$. The first few of these coefficients are
\begin{align*}
\cF_{01} &= u_{2}, \\
\cF_{02} &= -3  u_{1}^{2} - u_{111} + \frac{3}{2}  u_{12} + u_{3}, \\
\cF_{03} &= -8  u_{1} u_{11} - 4  u_{1} u_{2} - \frac{4}{3}  u_{1111} + \frac{4}{3}  u_{13} + u_{22} + u_{4}, \\
\cF_{04} &= - 5  u_{11}^2 - \frac{20}{3} u_{1} u_{111} - 10  u_{1} u_{12} - 5  u_{11} u_{2} - \frac{5}{4} u_{2}^2 + \frac{10}{3}  u_{1} u_{3} - u_{11111} \\
&\qquad - \frac{5}{3} u_{1112} + \frac{5}{4}  u_{122} + \frac{5}{4}  u_{14} + \frac{5}{3} u_{23} + u_{5},  \\
&\ \vdots 
\end{align*}
where once again we use the subscript $i$ rather than $t_i$ to denote partial derivatives of $u$. We see that the flows corresponding to even times are trivial. In the odd orders we find the pKdV equations,
\begin{align*}
u_{2} &= 0, \\
u_{3} &= 3 u_{1}^2 + u_{111}, \\
u_{4} &= 0, \\
u_{5} &= 10 u_{1}^3 + 5 u_{11}^2 + 10 u_{1} u_{111} + u_{11111}, \\
&\ \vdots
\end{align*}
For $i \geq 1$, the equations $\cF_{ij} = 0$ are consequences of these equations.

\subsubsection{Pluri-Lagrangian structure}

A pluri-Lagrangian description of Equation \eqref{H1} was found in \cite{lobb2009lagrangian}, the Lagrange function itself goes back to \cite{capel1991complete}. It reads
\[L(V,V_i,V_j,V_{ij},\alpha_i,\alpha_j) = V(V_i -V_j) - (\alpha_i - \alpha_j) \log(V_i-V_j) . \]
Using Lemma \ref{lemma-nulllag1}, we see that this Lagrangian is equivalent to (denoted with $=$ by abuse of notation)
\[ L(V,V_i,V_j,V_{ij},\alpha_i,\alpha_j) = \frac{1}{2}(V - V_{ij})(V_i - V_j) + (\alpha_i - \alpha_j) \log(V_i-V_j) . \]
In terms of $U$ and $\lambda$ it is (up to a constant)
\begin{align*}
L(U,U_i,U_j,U_{ij},\lambda_i,\lambda_j) &= \frac{1}{2} \left( U - U_{ij} - \lambda_i^{-1} - \lambda_j^{-1} \right)\left( U_i - U_j + \lambda_i^{-1} - \lambda_j^{-1} \right) \\
&\qquad + \left( \lambda_i^{-2} - \lambda_j^{-2} \right) \log\!\left(1 + \frac{U_i-U_j}{\lambda_i^{-1} - \lambda_j^{-1}} \right). 
\end{align*}
\begin{lemma}\label{lemma-nulllag2}
$L_0(U,U_i,U_j,U_{ij},\alpha_i,\alpha_j) = (\lambda_i^{-1} + \lambda_j^{-1})(U_i - U_j) + (\lambda_i^{-1} - \lambda_j^{-1})(U - U_{ij})$ is a null Lagrangian.
\end{lemma}
\begin{proof}
Consider the discrete 1-form $\eta$ defined by $\eta(U,U_i,\lambda_i) = \lambda_i^{-1} (U +  U_i)$ and, dictated by symmetry, $\eta(U_i,U,\lambda_i) = -\lambda_i^{-1} (U + U_i)$. Its discrete exterior derivative is
\[ 
\Delta \eta(U,U_i,U_{ij},U_j,\lambda_i,\lambda_j) 
= \frac{U+U_i}{\lambda_i} + \frac{U_i+U_{ij}}{\lambda_j} - \frac{U_{ij}+U_j}{\lambda_i} - \frac{U_j+U}{\lambda_j} \\
= L_0. \qedhere
\]
\end{proof}

Lemma \ref{lemma-nulllag2} implies that $L$ is equivalent to
\begin{equation}\label{L-lpKdV}
\begin{split}
L(U,U_i,U_j,U_{ij},\lambda_i,\lambda_j) &= \frac{1}{2} \left( U - U_{ij} - 2 \lambda_i^{-1} - 2 \lambda_j^{-1} \right) (U_i - U_j) \\
&\qquad + \left( \lambda_i^{-2} - \lambda_j^{-2} \right) \log\!\left(1 + \frac{U_i-U_j}{\lambda_i^{-1} - \lambda_j^{-1}} \right) .
\end{split}
\end{equation}
To see why this Lagrangian is preferable, do a first order Taylor expansion of the logarithm and admire the cancellation. Thanks to this cancellation we avoid terms of non-positive order in the series expansion.

Applying the Miwa correspondence \eqref{miwa} with $c = -2$, a Taylor expansion, and the Euler-Maclaurin formula to the Lagrangian \eqref{L-lpKdV}, we obtain a power series
\[\cL_\miwa([u],\lambda_1, \lambda_2) = \sum_{ij=1}^\infty \frac{4 (-1)^{i+j}}{ij} \cL_{ij}[u] \lambda_1^i \lambda_2^j, \]
whose coefficients define a continuous pluri-Lagrangian 2-form for the KdV hierarchy. The first row of coefficients reads:
\begin{align*}
\cL_{12} &= u_{1} u_{2}
\\
\cL_{13} &=  -2 u_{1}^{3} - u_{1} u_{111} + \frac{3}{4} u_{2}^{2} + u_{1} u_{3}
\\
\cL_{14} &= -4 u_{1}^{2} u_{2} - \frac{4}{3} u_{1} u_{112} - \frac{2}{3} u_{11} u_{12} - \frac{2}{3} u_{111} u_{2} + \frac{4}{3} u_{2} u_{3} + u_{1} u_{4}
\\
\cL_{15} &= \frac{10}{3} u_{1} u_{11}^{2} - \frac{5}{2} u_{1} u_{2}^{2} - \frac{10}{3} u_{1}^{2} u_{3} + \frac{5}{9} u_{11} u_{1111} + \frac{1}{9} u_{1} u_{11111} - \frac{10}{9} u_{1} u_{113} - \frac{5}{6} u_{12}^{2} \\
&\quad - \frac{5}{12} u_{1} u_{122} - \frac{5}{9} u_{11} u_{13} - \frac{5}{6} u_{112} u_{2} - \frac{5}{12} u_{11} u_{22} - \frac{5}{9} u_{111} u_{3} + \frac{5}{9} u_{3}^{2} + \frac{5}{4} u_{2} u_{4} + u_{1} u_{5} \\
&\ \vdots
\end{align*}
Note that we can get rid of the alien derivatives in each $\cL_{1j}$ by adding a total derivative $\D_{t_1} c_j$ and discarding terms that have a double zero on solutions. To make sure we get an equivalent Lagrangian 2-form, we also add $\D_{t_i} c_j$ to the coefficients $\cL_{ij}$, which amounts to adding the closed form $\d \!\left( \sum_j c_j \d t_j \right)$ to $\cL$. In this particular example we take
\begin{align*}
\sum_j c_j \d t_j &= 
\left(\frac{4}{3} u_{1} u_{12} - \frac{2}{3} u_{11} u_{2} \right) \d t_4 \\
&\quad + \left( \frac{10}{3} u_{1}^{2} u_{11} - \frac{4}{9} u_{11} u_{111} - \frac{1}{9} u_{1} u_{1111} + \frac{10}{9} u_{1} u_{13} + \frac{5}{12} u_{1} u_{22} - \frac{5}{9} u_{11} u_{3} \right) \d t_5  \\
&\quad + \ldots .
\end{align*}
Now that we have disposed of the alien derivatives in the $\cL_{1j}$, we can use Theorem \ref{thm-alien} to eliminate the remaining alien derivatives in all other $\cL_{ij}$. For $i < j \leq 5$, the coefficients obtained this way are displayed in Table \ref{table-H1-clean}.

This procedure of eliminating alien derivatives can be done algorithmically: 
\begin{enumerate}
	\item Eliminate products of alien derivatives by adding double zeros. Using the notation of Equation \eqref{fullel}, we replace
	\[ u_I u_J \mapsto u_I u_J - (u_I - f_I[u]) (u_J - f_J[u]) = u_I f_J[u] + u_J f_I[u] - f_I[u] f_J[u] .\]
	\item In each equivalence class of single terms modulo $t_1$-derivatives, choose a preferred representative (e.g.\@ minimize the order of the highest occurring $t_1$-derivative in the term). Then choose the $c_j$ such that all terms in the $\cL_{1j}$ are brought in this preferred form. This simplifies $\cL_{1j}$ by combining equivalent terms, but there is no guarantee that it will remove all alien derivatives.
\end{enumerate}
Although we do not have a proof that this algorithm eliminates all alien derivatives, it works for the examples considered here and for all other ABS equations that we have been able to construct the continuum limit of.

\begin{table}[t]
\begin{empheq}[box=\fbox]{align*}
\cL_{12} & =  u_{1} u_{2} 
\\
\cL_{13} & = -2 u_{1}^{3} - u_{1} u_{111} + u_{1} u_{3} 
\\
\cL_{14} & = u_{1} u_{4}
\\
\cL_{15} & = -5 u_{1}^{4} + 10 u_{1} u_{11}^{2} - u_{111}^{2} + u_{1} u_{5}
\\[1em]
\cL_{23} & =  -3 u_{1}^{2} u_{2} - u_{1} u_{112} + u_{11} u_{12} - u_{111} u_{2}
\\
\cL_{24} & = 0
\\
\cL_{25} & = -10 u_{1}^{3} u_{2} + 20 u_{1} u_{11} u_{12} - 5 u_{11}^{2} u_{2} - 10 u_{1} u_{111} u_{2} - 2 u_{111} u_{112} + 2 u_{1111} u_{12} \\
&\quad - u_{11111} u_{2}
\\[1em]
\cL_{34} & = 3 u_{1}^{2} u_{4} + u_{1} u_{114} - u_{11} u_{14} + u_{111} u_{4}
\\*
\cL_{35} & = 6 u_{1}^{5} - 15 u_{1}^{2} u_{11}^{2} + 20 u_{1}^{3} u_{111} - 10 u_{1}^{3} u_{3} + 7 u_{11}^{2} u_{111} + 6 u_{1} u_{111}^{2} - 12 u_{1} u_{11} u_{1111} \\
&\quad + 3 u_{1}^{2} u_{11111} + 20 u_{1} u_{11} u_{13} - 5 u_{11}^{2} u_{3} - 10 u_{1} u_{111} u_{3} + 3 u_{1}^{2} u_{5} - u_{1111}^{2} \\
&\quad + u_{111} u_{11111} - 2 u_{111} u_{113} + u_{1} u_{115} + 2 u_{1111} u_{13} - u_{11} u_{15} - u_{11111} u_{3} + u_{111} u_{5} 
\\[1em]
\cL_{45} & = -10 u_{1}^{3} u_{4} + 20 u_{1} u_{11} u_{14} - 5 u_{11}^{2} u_{4} - 10 u_{1} u_{111} u_{4} - 2 u_{111} u_{114} + 2 u_{1111} u_{14} \\
&\quad - u_{11111} u_{4}
\end{empheq}
\caption{Coefficients $\cL_{ij}$ for H1, after eliminating alien derivatives.}
\label{table-H1-clean}
\end{table}

Note that the equations $u_{2i} = 0$ restrict the dynamics to a space of half the dimension. We can also restrict the pluri-Lagrangian formulation to this space:
\[ \cL = \sum_{i,j} \cL_{2i+1,2j+1} \,\d t_{2i+1} \wedge \d t_{2j+1} \]
is a pluri-Lagrangian $2$-form for the hierarchy of non-trivial pKdV equations,
\begin{align*}
u_{3} &= 3 u_{1}^2 + u_{111}, \\
u_{5} &= 10 u_{1}^3 + 5 u_{11}^2 + 10 u_{1} u_{111} + u_{11111}, \\
&\ \vdots
\end{align*}
On the level of equations we could have restricted to the odd-numbered coordinates $t_1,t_3,\ldots$ from the beginning. However, on the level of Lagrangians we need to consider the even-numbered coordinates as well, at least in the theoretical arguments, because otherwise there is no interpretation for the coefficients of $\lambda_1^{2i} \lambda_2^{2j}$ in the power series $\cL_\miwa$, which usually do not vanish before eliminating alien derivatives.

\subsubsection{The double continuum limit of Wiersma and Capel}\label{sec-wc}

In \cite{wiersma1987lattice} Wiersma and Capel presented a continuum limit of the lpKdV equation
\begin{equation}\label{dkdv}
(\mu_1 + \mu_2 + U_{12} - U)(\mu_1 - \mu_2  + U_1 - U_2) = \mu_1^2 - \mu_2^2 ,
\end{equation}
which is equivalent to equation \eqref{lpKdV} under the transformation $\mu_i = \lambda_i^{-1}$. Their procedure consists of two steps. First they obtain a hierarchy of differential-difference equations. A second continuum limit, applied to any single equation of this hierarchy, then yields the potential KdV hierarchy. Some ideas concerning this limit procedure were already developed in \cite{nijhoff1983direct, quispel1984linear}. Here we will summarize both limits in one step.

The limit procedure from \cite{wiersma1987lattice} uses the lattice parameters $\nu = \mu_1 - \mu_2$ and $\mu_1$ itself, and skew lattice coordinates:
\[ V(n,m) = U(n-m,m) . \]
Consider an interpolating function $u$. If
\[ V(n,m) = U(n-m,m) = u (t_1, t_3, t_5, \ldots), \]
then after the double limit of \cite{wiersma1987lattice}, lattice shifts correspond to multi-time shifts as follows:
\[ V_1 = U_1 = u \bigg( t_1 - \frac{2}{\mu_1}, t_3 - \frac{2}{3 \mu_1^3}, t_5 - \frac{2}{5 \mu_1^5}, \ldots \bigg) \]
and
\begin{align*}
V_2 = U_{-1,2} = u\bigg( & t_1 + \nu \frac{2}{\mu_1^2} - \frac{\nu^2}{2} \frac{2}{\mu_1^3} + \frac{\nu^3}{3} \frac{2}{p^4} - \ldots, t_3 + \nu \frac{2}{\mu_1^4} - \frac{\nu^2}{2} \frac{4}{\mu_1^5} + \frac{\nu^3}{3} \frac{20}{3\mu_1^6} - \ldots, \\
& t_5 + \nu \frac{2}{\mu_1^6} - \frac{\nu^2}{2} \frac{6}{\mu_1^7} + \frac{\nu^3}{3} \frac{14}{\mu_1^8} - \ldots , \ldots \bigg).
\end{align*}
The series occurring here can be recognized as Taylor expansions:
\begin{align*}
V_2 = u\bigg( t_1 -  \left( \frac{2}{\mu_1+\nu} - \frac{2}{\mu_1} \right),
\,& t_3 - \frac{1}{3}\left( \frac{2}{(\mu_1+\nu)^3} - \frac{2}{\mu_1^3} \right), \\
& t_5 - \frac{1}{5}\left( \frac{2}{(\mu_1+\nu)^5} - \frac{2}{\mu_1^5} \right), \ldots \bigg) .
\end{align*}
Going back to the straight lattice coordinates and the original lattice parameters $\mu_1$ and $\mu_2 = \mu_1 + \nu$, we find
\begin{equation}\label{oddmiwa}
\begin{split}
U_2 & = V_{12} = u \bigg( t_1 - \frac{2}{\mu_2}, t_3 - \frac{2}{3\mu_2^3}, t_5 - \frac{2}{5\mu_2^5}, \ldots \bigg) , \\
U_1  & = V_{1\phantom{2}} = u \bigg( t_1 - \frac{2}{\mu_1}, t_3 - \frac{2}{3 \mu_1^3}, t_5 - \frac{2}{5 \mu_1^5}, \ldots \bigg) .
\end{split}
\end{equation}
Hence the end result of the double limit of Wiersma and Capel is the same as the limit we obtain using the odd-numbered Miwa variables only.

\subsection[Cross-ratio equation]{Cross-ratio equation (Q1$_{\delta=0}$)}

\subsubsection{Equation}

Consider equation $Q1$ from the ABS list \cite{adler2003classification}, with parameter $\delta = 0$, 
\begin{equation}\label{Q1}
\alpha_1 (V_2 - V)(V_{12} - V_1) - \alpha_2 (V_1 - V)(V_{12} - V_2) = 0 .
\end{equation}
It is also known as the \emph{cross-ratio equation} \cite{bobenko2002integrable,nijhoff1995discrete} and as the \emph{lattice Schwarzian KdV equation} \cite[Chapter 3]{hietarinta2016discrete}. As before, we would like to view Equation \eqref{Q1} as a consistent numerical discretization of some differential equation. To achieve this, we identify $\alpha_1 = \lambda_1^2$ and $\alpha_2 = \lambda_2^2$. Then Equation \eqref{H1} is equivalent to
\begin{equation}\label{lSKdV}
\frac{V_1 - V}{\lambda_1}\frac{V_{12} - V_2}{\lambda_1}  - \frac{V_2 - V}{\lambda_2} \frac{V_{12} - V_1}{\lambda_2} = 0 .
\end{equation}
If we identify $V = v(t_1,\ldots)$, $V_i = v(t_1 + \lambda_i,\ldots)$, etc., then the leading order expansion yields $v_{t_1}^2 - v_{t_1}^2 = 0$. This is a tautological equation, as desired. Hence in this case there is no need for an additional change of variables. 

Once more we use the the Miwa correspondence \eqref{miwa} with $c = -2$. A Taylor expansion of \eqref{lSKdV} yields
\[ \sum_{i,j} \frac{4(-1)^{i+j}}{ij} \cF_{ij}[v] \lambda_1^i \lambda_2^j = 0 \]
with
\begin{align*}
\cF_{01} &= v_{1} v_{2}, \\
\cF_{02} &= \frac{3}{2}  v_{11}^{2} - v_{1} v_{111} +
\frac{3}{2}  v_{1} v_{12} + \frac{3}{2}  v_{11} v_{2} + \frac{3}{8} v_{2}^{2} + v_{1} v_{3}, \\
\cF_{03} &= \frac{8}{3}  v_{11} v_{111} - \frac{4}{3}
 v_{1} v_{1111} + 4  v_{11} v_{12} + \frac{4}{3}  v_{1} v_{13} + \frac{4}{3}  v_{111} v_{2} + 2  v_{12} v_{2} + v_{1} v_{22} + \frac{4}{3}  v_{11} v_{3} \\
&\qquad + \frac{2}{3}  v_{2} v_{3} + v_{1} v_{4}, \\
\cF_{04} &= -\frac{10}{9}  v_{111}^{2} - \frac{5}{3}  v_{11} v_{1111} + v_{1} v_{11111} + \frac{5}{3}  v_{1} v_{1112} - 5  v_{11} v_{112} - \frac{10}{3} v_{111} v_{12} - \frac{5}{2}  v_{12}^{2} \\
&\qquad - \frac{5}{4} v_{1} v_{122} - \frac{10}{3}  v_{11} v_{13} - \frac{5}{4}  v_{1} v_{14} - \frac{5}{6}  v_{1111} v_{2} - \frac{5}{2}  v_{112} v_{2} - \frac{5}{3}  v_{13} v_{2} - \frac{5}{4}  v_{11} v_{22} \\
&\qquad - \frac{5}{8} v_{2} v_{22} - \frac{5}{3}  v_{1} v_{23} - \frac{10}{9}  v_{111} v_{3} - \frac{5}{3}  v_{12} v_{3} - \frac{5}{18}  v_{3}^{2} - \frac{5}{4}  v_{11} v_{4} - \frac{5}{8}  v_{2} v_{4} - v_{1} v_{5},  \\
&\ \vdots 
\end{align*}
We assume that $v_{1} \neq 0$. Then we see that the flows corresponding to even times are trivial. In the odd orders we find the hierarchy of Schwarzian KdV equations,
\begin{equation}\label{SKdVhier}
\begin{split}
v_{2} &= 0, \\
\frac{v_{3}}{v_{1}} &= -\frac{3 v_{11}^{2}}{2 v_{1}^{2}} + \frac{v_{111}}{v_{1}}
, \\
v_{4} &= 0, \\
\frac{v_{5}}{v_{1}} &= -\frac{45  v_{11}^{4}}{8  v_{1}^{4}} + \frac{25  v_{11}^{2} v_{111}}{2  v_{1}^{3}} - \frac{5  v_{111}^{2}}{2  v_{1}^{2}} - \frac{5  v_{11} v_{1111}}{v_{1}^{2}} + \frac{v_{11111}}{v_{1}}, \\
&\ \vdots
\end{split}
\end{equation}
For $i \geq 1$, the equations $\cF_{ij} = 0$ are differential consequences of these equations.

\subsubsection{Pluri-Lagrangian structure}

\begin{table}[p]
\begin{empheq}[box=\fbox]{align*}
\cL_{12} & = -\frac{{v_{2}}}{4 {v_{1}}} \\
\cL_{13} & = \frac{{v_{11}}^{2}}{4 {v_{1}}^{2}} - \frac{{v_{3}}}{4 {v_{1}}} \\
\cL_{14} & =  -\frac{{v_{4}}}{4 {v_{1}}} \\
\cL_{15} & =  \frac{3 {v_{11}}^{4}}{16 {v_{1}}^{4}} - \frac{{v_{111}}^{2}}{4 {v_{1}}^{2}} - \frac{{v_{5}}}{4 {v_{1}}}
\\[.8em]
\cL_{23} & =  \frac{{v_{11}} {v_{12}}}{2 {v_{1}}^{2}} + \frac{{v_{11}}^{2} {v_{2}}}{8 {v_{1}}^{3}} - \frac{{v_{111}} {v_{2}}}{4 {v_{1}}^{2}} \\
\cL_{24} & = 0 \\
\cL_{25} & =  -\frac{{v_{111}} {v_{112}}}{2 {v_{1}}^{2}} + \frac{3 {v_{11}}^{3} {v_{12}}}{4 {v_{1}}^{4}} - \frac{{v_{11}} {v_{111}} {v_{12}}}{{v_{1}}^{3}} + \frac{{v_{1111}} {v_{12}}}{2 {v_{1}}^{2}} + \frac{27 {v_{11}}^{4} {v_{2}}}{32 {v_{1}}^{5}} 
 - \frac{17 {v_{11}}^{2} {v_{111}} {v_{2}}}{8 {v_{1}}^{4}} \\
&\quad + \frac{7 {v_{111}}^{2} {v_{2}}}{8 {v_{1}}^{3}} + \frac{3 {v_{11}} {v_{1111}} {v_{2}}}{4 {v_{1}}^{3}} - \frac{{v_{11111}} {v_{2}}}{4 {v_{1}}^{2}}
\\[.8em]
\cL_{34} & =  -\frac{{v_{11}} {v_{14}}}{2 {v_{1}}^{2}} - \frac{{v_{11}}^{2} {v_{4}}}{8 {v_{1}}^{3}} + \frac{{v_{111}} {v_{4}}}{4 {v_{1}}^{2}}
\\
\cL_{35} & = \frac{45 {v_{11}}^{6}}{64 {v_{1}}^{6}} - \frac{57 {v_{11}}^{4} {v_{111}}}{32 {v_{1}}^{5}} + \frac{19 {v_{11}}^{2} {v_{111}}^{2}}{16 {v_{1}}^{4}} - \frac{7 {v_{111}}^{3}}{8 {v_{1}}^{3}} + \frac{3 {v_{11}}^{3} {v_{1111}}}{8 {v_{1}}^{4}} + \frac{3 {v_{11}} {v_{111}} {v_{1111}}}{4 {v_{1}}^{3}} \\
&\quad - \frac{{v_{1111}}^{2}}{4 {v_{1}}^{2}} - \frac{3 {v_{11}}^{2} {v_{11111}}}{8 {v_{1}}^{3}} + \frac{{v_{111}} {v_{11111}}}{4 {v_{1}}^{2}} - \frac{{v_{111}} {v_{113}}}{2 {v_{1}}^{2}} + \frac{3 {v_{11}}^{3} {v_{13}}}{4 {v_{1}}^{4}} - \frac{{v_{11}} {v_{111}} {v_{13}}}{{v_{1}}^{3}} \\
&\quad + \frac{{v_{1111}} {v_{13}}}{2 {v_{1}}^{2}} - \frac{{v_{11}} {v_{15}}}{2 {v_{1}}^{2}} + \frac{27 {v_{11}}^{4} {v_{3}}}{32 {v_{1}}^{5}} - \frac{17 {v_{11}}^{2} {v_{111}} {v_{3}}}{8 {v_{1}}^{4}} + \frac{7 {v_{111}}^{2} {v_{3}}}{8 {v_{1}}^{3}} + \frac{3 {v_{11}} {v_{1111}} {v_{3}}}{4 {v_{1}}^{3}} \\
&\quad - \frac{{v_{11111}} {v_{3}}}{4 {v_{1}}^{2}} - \frac{{v_{11}}^{2} {v_{5}}}{8 {v_{1}}^{3}} + \frac{{v_{111}} {v_{5}}}{4 {v_{1}}^{2}}
\\[.8em]
\cL_{45} & = -\frac{{v_{111}} {v_{114}}}{2 {v_{1}}^{2}} + \frac{3 {v_{11}}^{3} {v_{14}}}{4 {v_{1}}^{4}} - \frac{{v_{11}} {v_{111}} {v_{14}}}{{v_{1}}^{3}} + \frac{{v_{1111}} {v_{14}}}{2 {v_{1}}^{2}} + \frac{27 {v_{11}}^{4} {v_{4}}}{32 {v_{1}}^{5}} - \frac{17 {v_{11}}^{2} {v_{111}} {v_{4}}}{8 {v_{1}}^{4}} \\
&\quad + \frac{7 {v_{111}}^{2} {v_{4}}}{8 {v_{1}}^{3}} + \frac{3 {v_{11}} {v_{1111}} {v_{4}}}{4 {v_{1}}^{3}} - \frac{{v_{11111}} {v_{4}}}{4 {v_{1}}^{2}} 
\end{empheq}
\caption{Coefficients $\cL_{ij}$ for Q1$_{\delta=0}$, after eliminating alien derivatives.}
\label{table-Q1-clean}
\end{table}

A Pluri-Lagrangian description of Equation \eqref{Q1} was found in \cite{lobb2009lagrangian}
\begin{equation}\label{Q1lag}
L = \alpha_i \log(V - V_i) - \alpha_j \log(V - V_i)  - (\alpha_i - \alpha_j) \log(V_i-V_j) ,
\end{equation}
which is equivalent to
\[
L = \lambda_i^2 \log\!\left( \frac{V - V_i}{\lambda_i} \right) - \lambda_j^2 \log\!\left( \frac{V - V_j}{\lambda_j} \right)  - (\lambda_i^2 - \lambda_j^2) \log\!\left( \frac{V_i - V_j}{\lambda_i - \lambda_j} \right) .
\]
Each term of the series $\cL_\miwa$ constructed form this discrete Lagrangian contains strictly positive powers of both $\lambda_i$ and $\lambda_j$. Thus by Theorem \ref{thm} we can identify the coefficients of this power series with the coefficients of a pluri-Lagrangian $2$-form. Some of these coefficients, after eliminating alien derivatives, are given in Table \ref{table-Q1-clean}.

Again we can restrict the pluri-Lagrangian formulation to a space of half the dimension and consider
\[ \cL = \sum_{i,j} \cL_{2i+1,2j+1} \,\d t_{2i+1} \wedge \d t_{2j+1} \]
as a pluri-Lagrangian $2$-form for the non-trivial equations of the SKdV hierarchy.

\subsubsection{The generating PDE of Nijhoff, Hone, and Joshi}

Nijhoff, Hone, and Joshi \cite{nijhoff2000schwarzian} introduced a non-autonomous PDE for a function $\mathtt{z}_{n,m}(t,s)$ depending on a pair of continuous variables $(s,t)$, and a pair of parameters $(m,n)$. They noted that the flow of this PDE in continuous $(s,t)$-coordinates commutes with the difference equations
\begin{equation}\label{zquad}
\frac{ (\mathtt{z}_{n,m} - \mathtt{z}_{n+1,m}) (\mathtt{z}_{n,m+1} - \mathtt{z}_{n+1,m+1}) }{ (\mathtt{z}_{n,m} - \mathtt{z}_{n,m+1}) (\mathtt{z}_{n+1,m} - \mathtt{z}_{n+1,m+1} ) } = \frac{s}{t} .
\end{equation}
Equation \eqref{zquad} is nothing but equation Q1$_{\delta=0}$. Hence it is possible to switch between the continuous and discrete picture by reversing the roles of parameters and independent variables. 

The main feature of the PDE in question is that it generates the SKdV hierarchy%
\footnote{Note that there is an error in the second SKdV equation as stated in \cite{nijhoff2000schwarzian}: the Lagrangian is missing the term $-\mathtt{z}_{x_2}^2 / \mathtt{z}_{x_1}^2$ at the corresponding order and in the equation itself the factor 2 of the first term should be removed.}
through the identification
\begin{align*}
\mathtt{z}_{n,m}(t,s) &= v \bigg( x_1 + \frac{2n}{t^{\frac{1}{2}}} + \frac{2m}{s^{\frac{1}{2}}}, x_3 + \frac{2n}{3 t^{\frac{3}{2}}} + \frac{2m}{3 s^{\frac{3}{2}}}, \ldots, \\
&\hspace{1cm} x_{2j+1} + \frac{2n}{(2j+1) t^{\frac{2j+1}{2}}} + \frac{2m}{(2j+1) s^{\frac{2j+1}{2}}}, \ldots  \bigg) . 
\end{align*}
Because of this it has become known as the \emph{generating PDE} \cite{lobb2009lagrangian,xenitidis2011lagrangian} for the SKdV hierarchy. Renaming the parameters $t = \lambda_1^2$ and $s = \lambda_2^2$ we obtain once again the odd order Miwa shifts
\begin{align*}
\mathtt{z}_{n,m} &= v \bigg( x_1 + \frac{2n}{\lambda_1} + \frac{2m}{\lambda_2}, x_3 + \frac{2n}{3 \lambda_1^3} + \frac{2m}{3 \lambda_2^3}, \ldots, \\
 &\hspace{1cm} x_{2j+1} + \frac{2n}{(2j+1) \lambda_1^{2j+1}} + \frac{2m}{(2j+1) \lambda_2^{2j+1}}, \ldots  \bigg) .
\end{align*}
Hence our continuum limit of Q1$_{\delta = 0}$ is implicitly present in \cite{nijhoff2000schwarzian}. The relation between the (non-autonomous) generating PDE, the quad equation, and the hierarchy of (autonomous) PDEs is illustrated in Figure \ref{fig-diagram}.

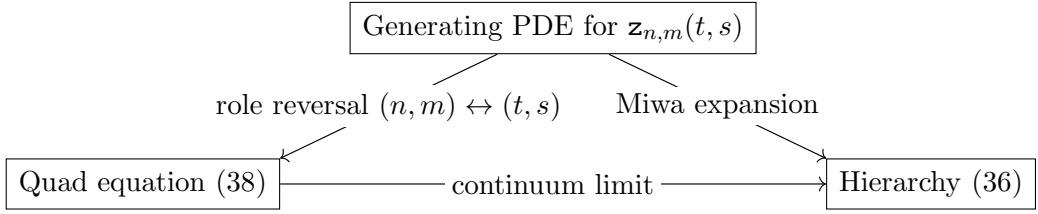
\begin{figure}[t]
\begin{tikzcd}[row sep = huge,cells={nodes={draw}}]
& \text{Generating PDE for } \mathtt{z}_{n,m}(t,s)
\arrow[dl, "\text{\normalsize role reversal $(n,m) \leftrightarrow (t,s)$} " description, end anchor=north east]
 \arrow[dr, "\text{\normalsize Miwa expansion}" description, end anchor=north west] &\\
\text{Quad\ equation \eqref{zquad}} \arrow[rr, "\text{\normalsize continuum limit}" description] & & \text{Hierarchy \eqref{SKdVhier}}
\end{tikzcd}
\caption{The continuum limit is compatible with the relations found in \cite{nijhoff2000schwarzian}}
\label{fig-diagram}
\end{figure}

\section{Conclusion}

We have presented a method to perform continuum limits of discrete pluri-Lagrangian systems. In this approach, a single (parameter-dependent) discrete equation produces a full hierarchy of differential equations, and the pluri-Lagrangian structure is carried over from the discrete system to the continuous one. The continuum limit produces 2-forms whose coefficients contain alien derivatives, but in our main examples these could be eliminated afterwards. Whether there exists a different continuum limit procedure that immediately delivers a native result, is an open question.

Although our method of taking the continuum limit can be stated in a general way, it can only be executed if we can find a form of the discrete equation and its Lagrangian that allows a suitable Taylor expansion in the parameters. Finding such a form is a non-trivial task. We solved this on a case-by-case basis for a few important examples. In a subsequent paper \cite{vermeeren2018variational} we will discuss many more examples, including all ABS equations of type Q.

\appendix
\section{The variational bicomplex}
\label{appendix}

This is a minimal introduction to the variational bicomplex. Much more on this topic can be found for example in \cite[Chapter 19]{dickey2003soliton}.

Our starting point is the idea that the exterior derivative can be split into a vertical part $\delta$ and a horizontal part $\d$. An $(a,b)$-form is a differential $(a+b)$-form structured as
\[ \omega^{a,b} = \sum f_{I_1,\ldots,I_a}([u]) \, \delta u_{I_1} \wedge \ldots \wedge \delta u_{I_a} \wedge \d t_{j_1} \ldots \wedge \d t_{j_b} . \]
Some authors use contact forms instead of the $\delta u_I$, see for example \cite{anderson1992introduction}. We denote the space of $(a,b)$-forms by $\Omega^{(a,b)} \subset \Omega^{a+b}$. These spaces are related to each other by $\d$ and $\delta$ as in the following diagram:
\[
\begin{tikzcd}
\vdots
&\vdots
&
&\vdots
&\vdots
\\
\Omega^{(2,0)} \arrow[r, "\d"] \arrow[u, "\delta"]
& \Omega^{(2,1)} \arrow[r, "\d"] \arrow[u, "\delta"]
& \ldots \arrow[r, "\d"]
& \Omega^{(2,n-1)} \arrow[r, "\d"] \arrow[u, "\delta"]
& \Omega^{(2,n)} \arrow[u, "\delta"]
\\
\Omega^{(1,0)} \arrow[r, "\d"] \arrow[u, "\delta"]
& \Omega^{(1,1)} \arrow[r, "\d"] \arrow[u, "\delta"]
& \ldots \arrow[r, "\d"]
& \Omega^{(1,n-1)} \arrow[r, "\d"] \arrow[u, "\delta"]
& \Omega^{(1,n)} \arrow[u, "\delta"]
\\
\Omega^{(0,0)} \arrow[r, "\d"] \arrow[u, "\delta"]
& \Omega^{(0,1)} \arrow[r, "\d"] \arrow[u, "\delta"]
& \ldots \arrow[r, "\d"]
& \Omega^{(0,n-1)} \arrow[r, "\d"] \arrow[u, "\delta"]
& \Omega^{(0,n)} \arrow[u, "\delta"]
\end{tikzcd}
\]

The horizontal and vertical exterior derivatives are characterized by the anti-derivation property,
\begin{align*}
\d \left( \omega_1^{p_1,q_1} \wedge \omega_2^{p_2,q_2} \right)
&= \d \omega_1^{p_1,q_1} \wedge \omega_2^{p_2,q_2} + (-1)^{p_1+q_1} \, \omega_1^{p_1,q_1} \wedge \d \omega_2^{p_2,q_2}  ,
\\
\delta \left( \omega_1^{p_1,q_1} \wedge \omega_2^{p_2,q_2} \right) 
&= \delta \omega_1^{p_1,q_1} \wedge \omega_2^{p_2,q_2} + (-1)^{p_1+q_1} \, \omega_1^{p_1,q_1} \wedge \delta \omega_2^{p_2,q_2} ,
\end{align*}
and by the way they act on $(0,0)$-forms, and basic $(1,0)$ and $(0,1)$-forms:
\[
\begin{tabular}{lcl}
$\displaystyle \d f = \sum_j (\D_{t_j} f) \,\d t_j $
&&
$\displaystyle \delta f = \sum_I \der{f}{u_I} \delta u_{I}$
\\
$\displaystyle \d (\delta u_I )= - \sum_j \delta u_{Ij} \wedge \d t_j$
& \hspace{2cm} &
$\delta(\delta u_I) = 0$
\\
$\d(\d t_j) = 0$ 
&&
$\delta (\d t_j) = 0$
\end{tabular}
\]
One can verify that $\d + \delta: \Omega^a \rightarrow \Omega^{a+1}$ is the usual exterior derivative and that
\[ \delta^2 = \d^2 = \delta \d + \d \delta = 0 . \]
Furthermore, for any (generalized) vector field $V$, there holds
\[ \iota_{V} \d + \d \iota_{V} = 0. \]
The derivative $\D_{t_i}$ acts on elementary $(1,0)$ forms as $\D_{t_i} \delta u_I = \delta u_{It_i}$ and satisfies the Leibniz rule with respect to $\wedge$.

The variational principle $\delta \int_\Gamma \cL = 0$, for an $(0,d)$-form $\cL = \sum L_{1 \ldots d} \, \d t_1 \wedge \ldots \wedge \d t_d$, can be understood as follows. For every vertical vector field
$ V = v \frac{\partial}{\partial u} $
that vanishes on the boundary $\partial \Gamma$, its prolongation
\[ \pr V = \sum_I (\D_I v) \frac{\partial}{\partial u_I} \]
must satisfy
\[ \int_\Gamma \iota_{\pr V} \delta \cL = \int_\Gamma \sum \iota_{\pr V} \delta L_{1 \ldots d} \wedge \d t_1 \wedge \ldots \wedge \d t_d = 0 . \]

An important property of classical Lagrangian systems is that changing the Lagrangian by a full derivative (or divergence) does not affect the Euler-Lagrange equations. The following Proposition is the pluri-Lagrangian generalization of this property.

\begin{prop}
\label{prop-delta=d}
The field $u$ is a critical point of the action $\int_\Gamma \cL[u]$ for all $\Gamma$ if there exists a $(1,d-1)$-form $\Theta[u]$ such that $\delta \cL = \d \Theta$.
\end{prop}
\begin{proof}
Since the horizontal exterior derivative $\d$ anti-commutes with the interior product operator $\iota_{\pr V}$ for the prolongation of a vertical vector field $V$, it follows that for any variation $V$ of the field $u$ that is zero on the boundary of a manifold $\Gamma$:
\[ \int_\Gamma \iota_{\pr V} \delta \cL = -\int_\Gamma \d \left( \iota_{\pr V} \Theta \right) = -\int_{\partial \Gamma} \iota_{\pr V} \Theta = 0 . \qedhere \]
\end{proof}

\subsection*{Acknowledgments}

This research was supported by the DFG Collaborative Research Center TRR 109, `Discretization in Geometry and Dynamics.'

\bibliographystyle{abbrvnat_jis}
\bibliography{clim}

\end{document}